\documentclass{article}

\usepackage{microtype}
\usepackage{graphicx}
\usepackage{booktabs} 

\usepackage{hyperref}

\usepackage{xcolor}
\usepackage[preprint]{icml2026}
\usepackage{etoolbox}
\makeatletter
\patchcmd{\printAffiliationsAndNotice}
  {Correspondence to:}
  {Corresponding author:}
  {}{}
\makeatother

\usepackage{amsmath}
\usepackage{amssymb}
\usepackage{mathtools}
\usepackage{amsthm}
\usepackage{subcaption}
\usepackage{float}
\usepackage{placeins} 
\usepackage[capitalize,noabbrev]{cleveref}

\theoremstyle{plain}
\newtheorem{theorem}{Theorem}[section]

\newtheorem{lemma}[theorem]{Lemma}

\theoremstyle{definition}

\theoremstyle{remark}

\usepackage[most]{tcolorbox}

\newtcblisting{cmdbox}[1]{%
  enhanced,
  breakable,
  colback=white,
  colframe=black!55,
  boxrule=0.8pt,
  arc=2mm,
  left=2mm,right=2mm,top=1.5mm,bottom=1.5mm,
  colbacktitle=black!10,
  coltitle=black,
  fonttitle=\bfseries,
  title={#1},
  attach boxed title to top left={xshift=2mm,yshift=-1mm},
  boxed title style={
    sharp corners,
    boxrule=0.8pt,
    arc=2mm,
    left=2mm,right=2mm,top=0.8mm,bottom=0.8mm,
    colframe=black!55,
    colback=black!10
  },
  listing only,
  listing options={
    basicstyle=\ttfamily\small,
    columns=fullflexible,
    breaklines=true
  }
}

\usepackage[textsize=tiny]{todonotes}

\usepackage{array}       
\usepackage{booktabs}   
\usepackage{algorithm}

\usepackage{booktabs}
\usepackage{rotating}
\icmltitlerunning{StarSD: One-for-Many Speculative Decoding}

\begin{document}

\twocolumn[
\icmltitle{StarSD: One-for-Many Speculative Decoding}




\begin{icmlauthorlist}
\icmlauthor{Junhao He}{hku}
\icmlauthor{Feiran You}{hku}
\icmlauthor{Hongyang Du}{hku}
\end{icmlauthorlist}

\icmlaffiliation{hku}{Department of Electrical and Electronic Engineering, The University of Hong Kong, Hong Kong SAR, China}

\icmlcorrespondingauthor{Hongyang Du}{duhy@hku.hk}


\vskip 0.3in
]

\printAffiliationsAndNotice{}  

\begin{abstract}
Speculative decoding accelerates autoregressive generation by separating token proposal from verification, but most existing approaches are designed for single-node execution and do not scale well to multi-accelerator clusters used for serving modern Large Language Models (LLMs). We present StarSD, a one-for-many speculative decoding framework that uses a single draft model to serve multiple target models across distributed nodes via a star topology. StarSD decouples drafting and verification, enabling effective sharing of draft computation, and preventing distributed accelerators from remaining idle under bursty workloads. We provide a system-level analysis that characterizes when and why a single draft model can remain fully utilized by multiple verifiers, yielding predictable latency and utilization gains. Extensive experiments in real-world distributed inference settings demonstrate that StarSD simplifies deployment and supports flexible resource allocation across heterogeneous accelerators, while maintaining output quality. These results indicate that StarSD is a practical and scalable framework for bringing speculative decoding to modern cloud and edge inference infrastructures.
\end{abstract}


\section{Introduction}
Large Language Models (LLMs) have become a central interface for a wide range of applications, and their deployment increasingly involves serving large volumes of concurrent inference requests in shared accelerator clusters. Under such workloads, inference systems face sustained pressure on latency, throughput, and hardware efficiency, often operating near the limits of memory and compute. A key source of this pressure is \emph{autoregressive} decoding, where each output token requires a serial forward pass, directly constraining end-to-end latency and limiting achievable throughput.
This bottleneck has motivated many decoding acceleration methods that restructure generation to reduce the cost of sequential decoding. Among them, speculative decoding offers a \emph{lossless} acceleration approach by introducing a lightweight \emph{draft} model to propose multiple tokens that a large \emph{target} model verifies in parallel, preserving the target model’s output distribution while reducing decoding time \cite{pmlr-v202-leviathan23a,chen2023specsampling}.
Unlike batching or approximate decoding, this separation mechanism preserves exact generation semantics while exposing additional parallelism, which makes speculative decoding suitable for latency-sensitive serving. Building on this advantage, several methods are proposed to improve draft quality and reduce verification costs, thereby further increasing acceptance rates and achieving decoding speedups~\cite{li2024eagle,miao2024specinfer,li2024eagle2}.

However, most speculative decoding designs implicitly adopt a centralized execution mode, in which the draft and target models are co-located on a single accelerator, e.g., a Graphics Processing Unit (GPU).
This pattern is increasingly misaligned with production inference systems, as modern LLMs continue to grow in parameter scale and deployment workloads impose sustained memory pressure through large Key-Value (KV) caches, causing the combined footprint to approach or even exceed the capacity of a single accelerator. 
For instance, co-locating a Vicuna-7B~\cite{zheng2023judging} target model with the lightweight EAGLE draft model can leave little memory headroom~\cite{li2024eagle} for an RTX 4090 (24GB), restricting support for long-context serving.
Moreover, as LLMs expand into vertical domains and are increasingly deployed on resource-constrained edge devices, e.g., PCs, mobile phones, or NVIDIA Jetson, co-location becomes even less practical due to tighter memory and compute budgets.
Although substantial work has explored memory-efficient inference techniques~\cite{kwon2023efficient,sheng2023flexgen}, these methods do not remove the need to host the target-draft model pair on every accelerator, which constrains batching efficiency and scheduling flexibility.

Consequently, LLM inference systems require a spatial degree of freedom, enabling distributed placement that allows drafting and verification to be placed independently while cooperating tightly to share accelerator resources. In principle, this flexibility supports higher utilization by decoupling compute-intensive drafting from request-driven verification, and applies both within a single server across multiple accelerators and distributed clusters across machines. 
However, simply separating the draft and target models reveals an accelerator-level inefficiency due to a mismatch between the model inference pattern and the operating characteristics of modern accelerators. Specifically, the draft-then-verify workflow alternates between drafting, verification, and communication, while the target model verifies proposed tokens, the draft model remains idle.
This alternating execution introduces \emph{idle gaps} that limit end-to-end throughput.
In addition, disaggregation makes the draft execution \emph{bursty}, with short compute phases separated by millisecond-scale gaps.
Such bursty execution can prevent the draft model from staying in a warm compute state, increasing per-token latency and its variance compared to continuous execution.

Addressing these challenges requires a novel speculative decoding framework that can coordinate drafting and verification across machines while sustaining steady execution on modern accelerators. We propose StarSD, a one-for-many distributed speculative decoding framework where a single draft model concurrently assists multiple target instances.
StarSD multiplexes asynchronous returns through a global request buffer and executes requests in a work-conserving loop with per-instance state isolation, ensuring continuous utilization of the draft accelerator whenever work is available.
To analyze performance under concurrency, we develop an analytical system model that decouples \emph{progress}, measured as the expected number of accepted tokens per round, from \emph{time}, measured as the round duration. The model characterizes under-loaded and fully-loaded regimes and identifies the transition point at which idle gaps vanish.

StarSD yields two complementary benefits under multi-instance workloads.
By amortizing verification and communication return latency across multiple targets, it reduces idle gaps and increases aggregate throughput.
Furthermore, a continuous request stream keeps the draft side closer to a steady-state execution regime, improving draft-side per-token compute efficiency and stability.
These effects make speculative decoding more effective in disaggregated and memory-constrained deployments without changing model weights or the output distribution. Our contributions are summarized as

\begin{itemize}
  \item We propose StarSD in the one-for-many distributed speculative decoding setting, a framework that isolates per-target state and enables work-conserving draft execution in the presence of asynchronous verification returns from targets.
  \item We develop an analytical system model that separates per-iteration progress from per-iteration time, characterizes load regimes, and explains scaling behavior, saturation effects, and the tradeoff between per-target throughput and system-wide throughput.
  \item We evaluate StarSD across various deployments, including intra-machine multi-GPU, inter-machine multi-GPU, and heterogeneous platforms such as a Computer-Jetson setup. The results show that StarSD scales throughput with the number of target instances until saturation, with continuous drafting reducing per-token compute time and latency variance compared with bursty execution.
\end{itemize}


\section{Preliminary}
\label{sec:preliminary}

Speculative decoding speeds up autoregressive inference by splitting decoding into two roles, i.e., a lightweight \emph{draft} model proposes candidate continuations, and a large \emph{target} model verifies them.
Following the original speculative decoding formulation~\cite{pmlr-v202-leviathan23a}, we denote the target and draft model types by $\mathcal{M}_p$ and $\mathcal{M}_q$, and use $M_p$ and $M_q$ to refer to their deployed instances, respectively.
When StarSD serves $N$ target instances concurrently, we index them as $\{M_p^{(i)}\}_{i=1}^{N}$. In the single-instance case, we drop the index and write $M_p \equiv M_p^{(1)}$ for notational simplicity.

\textbf{Per-request primitive.}
Given a verified prefix, one speculative decoding step contains of two stages:
(i) \textbf{proposal}, where the draft instance $M_q$ generates candidate continuations from the prefix, and
(ii) \textbf{verification}, where the target instance $M_p$ verifies the proposals and returns an \emph{accepted prefix} defined as the longest prefix that passes verification.
Both models then advance to the accepted prefix and continue decoding.

\textbf{Tree-based proposals.}
SpecInfer~\cite{miao2024specinfer} popularizes \emph{tree-based} speculative proposals to better exploit accelerator parallelism.
Rather than proposing a single token chain, $M_q$ expands a token tree of depth $d$ with branching factor $k$, i.e., top-$k$ candidates per step, so that the target model can score many candidate continuations in parallel within a batched forward pass.
Following this design, advanced systems such as EAGLE~\cite{li2024eagle} adopt tree-style expansion to increase the likelihood of finding a long verifiable prefix while maintaining high verification throughput.
In this work, we also adopt the tree-based drafting and use the depth $d$ to parameterize the draft workload.

\textbf{Acceptance intuition.}
The number of accepted tokens mainly depends on how well $\mathcal{M}_q$ matches $\mathcal{M}_p$ under the current prefix. A smaller draft-target mismatch typically yields longer accepted prefixes.
Beyond architectural designs, several works explicitly align $\mathcal{M}_q$ to $\mathcal{M}_p$ via training, e.g., DistillSpec~\cite{zhou2023distillspec} and AdaSPEC~\cite{hu2025adaspec} use distillation tailored to speculative decoding, while Online Speculative Decoding ~\cite{liu2023online} reduces mismatch after deployment by periodically updating $\mathcal{M}_q$ using the corrections from $\mathcal{M}_p$.
EAGLE~\cite{li2024eagle,2025arXiv250301840L} complements these training-based approaches by conditioning $\mathcal{M}_q$ on additional signals from $\mathcal{M}_p$, i.e., last-layer hidden states, making proposals more target-aware.

\section{System Model}
We present the system model for StarSD, where one $M_q$ serves $N$ target instances $\{M_p^{(i)}\}_{i=1}^{N}$. The model characterizes system throughput as a function of $N$ and the per-request workload by capturing the expected accepted tokens and the iteration time under different load conditions.


\subsection{Iteration Definitions}
System iterations are indexed by $\gamma = \{1, \ldots, \gamma_{\mathrm{end}}\}$, where each iteration corresponds to one draft-side scheduling round, and $\gamma_{\mathrm{end}}$ denotes the total number of scheduling rounds executed during a decoding session. In the $\gamma$-th iteration, the draft instance $M_q$ processes one speculative request for each active target instance $M_p^{(i)}$ in a work-conserving order.
For a given request, $M_q$ generates a speculative draft with depth $d$, which is verified by the corresponding $M_p^{(i)}$, and the verification result is returned to $M_q$ before the next iteration.
The draft depth $d$ upper-bounds the per-iteration progress and determines the amount of work for each $M_p^{(i)}$. This abstraction allows us to analyze system progress and time cost on a per-iteration basis.

\subsection{Per-iteration Acceptance Length}
We first quantify the per-iteration progress made by all target instances in one iteration, measured by the number of tokens accepted by $\{M_p^{(i)}\}_{i=1}^{N}$.

\paragraph{Acceptance probability.}
We start from a token-level acceptance notion that captures the agreement between $\mathcal{M}_q$ and $\mathcal{M}_p$.
We consider that all target instances $\{M_p^{(i)}\}_{i=1}^{N}$ are instantiated from the same target model type $\mathcal{M}_p$, which is common in multi-session serving scenarios, e.g., LLM-based multi-agent systems \cite{NEURIPS2023_a3621ee9,chen2024agentverse,wu2024autogen}. 
The index $i$ distinguishes concurrent decoding sessions that may have different prefixes and thus different next-token distributions.
Conditioned on a prefix $y$, let the next-token distributions of target $M_p^{(i)}$ and draft $M_q$ be $p_i(\cdot\mid y)$ and $q(\cdot\mid y)$, respectively.
We define the token-level acceptance probability at prefix $y$ as
\begin{equation}
\begin{aligned}
\beta_i(y)
&\triangleq \sum_x \min\{p_i(x\mid y),\, q(x\mid y)\} \\
&= 1 - D_{\mathrm{LK}}\!\left(p_i(\cdot\mid y),\, q(\cdot\mid y)\right),
\end{aligned}
\label{eq:beta_cond}
\end{equation}
where $D_{\mathrm{LK}}(p,q)=\frac{1}{2}\sum_x |p(x)-q(x)|$ denotes the total variation distance between $p$ and $q$ \cite{pmlr-v202-leviathan23a}.
This quantity measures how likely a draft proposal from $M_q$ is accepted by $M_p^{(i)}$ under the given prefix $y$.


\paragraph{Expected accept length.}
We lift token-level acceptance to an iteration-level quantity by defining the accept length.
Let $l_i(d)\in\{0,1,\dots,d\}$ denote the accept length of $M_p^{(i)}$ in one iteration, i.e., the number of consecutive proposed tokens accepted by $M_p^{(i)}$.
This random variable represents the realized per-iteration progress for $M_p^{(i)}$.
To obtain the expected accept length, we use a standard tail-sum identity that converts $\mathbb{E}[l_i(d)]$ into a sum of acceptance events.

\begin{lemma}[Tail-sum identity for bounded accept length]
For any $l_i(d)\in\{0,1,\dots,d\}$, we have
\begin{equation}
\mathbb{E}\!\left[l_i(d)\right] = \sum_{j=1}^{d} \Pr\!\left(l_i(d) \ge j\right).
\label{eq:tail_sum_li}
\end{equation}
\end{lemma}
The term $\Pr \left(l_i\left(d\right)\ge j\right)$ is the probability that the first $j$ proposed tokens are all accepted, and their sum yields the expected accept length.

\begin{proof}
Since $l_i(d)$ is non-negative and bounded by $d$, it can be written as
$
l_i(d) = \sum_{j=1}^{d}\mathbf{1}\{l_i(d)\ge j\},
$
where $\mathbf{1}\{\cdot\}$ is the indicator function.
Taking the expectation on both sides and using the linearity of expectation gives
$
\mathbb{E}[l_i(d)] = \sum_{j=1}^{d}\mathbb{E}[\mathbf{1}\{l_i(d)\ge j\}]
= \sum_{j=1}^{d}\Pr(l_i(d)\ge j).
$
\end{proof}

\paragraph{Closed-form Expression.}
A common approximation in speculative decoding analysis is that, within one iteration, token acceptances are i.i.d. Bernoulli with an acceptance rate $\beta_i$ \cite{pmlr-v202-leviathan23a}. However, later drafts depend on earlier acceptances in practice, so the independence assumption is an analytical simplification \cite{chen2023specsampling}.
Under this approximation
\begin{equation}
\mathbb{E}\!\left[l_i(d)\right] \approx \sum_{j=1}^{d}\beta_i^{\,j}
= \frac{\beta_i\left(1-\beta_i^{\,d}\right)}{1-\beta_i}.
\label{eq:El_closed}
\end{equation}
This expression shows diminishing returns in $d$ when $\beta_i<1$, as the marginal gain scales with $\beta_i^{d+1}$.

Finally, the expected total accepted tokens across $\{M_p^{(i)}\}_{i=1}^{N}$ in one iteration is
\begin{equation}
\mathbb{E}[l]_\gamma \triangleq \sum_{i=1}^{N}\mathbb{E}\!\left[l_i(d)\right]
\approx \sum_{i=1}^{N}\frac{\beta_i\left(1-\beta_i^{\,d}\right)}{1-\beta_i}.
\label{eq:El_gamma}
\end{equation}
This quantity aggregates per-iteration progress and serves as the system progress term when interpreting throughput.

\subsection{Per-iteration Time Decomposition}
We turn to the \emph{time} component of throughput by characterizing the duration of a single iteration under concurrent service.
In each iteration, $M_q$ serves $N$ depth-$d$ requests, one from each active target instance $M_p^{(i)}$.
The service time for a single depth-$d$ request is
\begin{equation}
S(d) \triangleq d\,t_s ,
\label{eq:Sd}
\end{equation}
where $t_s$ is the time for $M_q$ to generate one layer of draft tokens, which remains nearly constant due to parallel execution within a modern accelerator.

After $M_q$ finishes serving a given $M_p^{(i)}$, the earliest time that $M_q$ can receive the next request from the same $M_p^{(i)}$ is determined by the return time as
\begin{equation}
Z(d) \triangleq t_c + t_v ,
\label{eq:Zd}
\end{equation}
where $t_c$ is the communication time in one iteration and $t_v$ is the time taken for $M_p^{(i)}$ to verify the received draft tokens. 

We consider $M_q$ is work-conserving, whenever the request buffer is non-empty, $M_q$ continuously serves pending requests according to a policy such as first-in-first-out.
Each target instance $M_p^{(i)}$ issues a new request only after completing verification of its previous one.
Under these settings, system load is fully determined by the number of concurrent target instances, without additional effects from scheduling or request reordering.

We next characterize the condition under which the draft instance $M_q$ remains continuously busy in StarSD pipeline.
\begin{lemma}[A sufficient non-idling condition for $M_q$]
\label{lem:busy}
If
\begin{equation}
Z(d) \le (N-1)\,S(d),
\label{eq:busy}
\end{equation}
then a work-conserving $M_q$ does not idle due to waiting for verification results from any target instance $M_p^{(i)}$.
\end{lemma}

\begin{proof}
After $M_q$ completes serving a request from some target instance $M_p^{(i)}$, the earliest time at which $M_p^{(i)}$ can issue its next request is $Z(d)$. During this interval, $M_q$ can serve at most one pending request from each of the remaining $N-1$ target instances, incurring a total service time of $(N-1)S(d)$. 
\end{proof}

\subsection{Per-iteration Throughput}
As shown in Fig.~\ref{fig:system pipeline fully loaded}, we consider two cases depending on whether \eqref{eq:busy} holds.

\textbf{Fully-loaded case.}
\begin{figure}[t]
    \centering
    \includegraphics[width=0.45\textwidth]{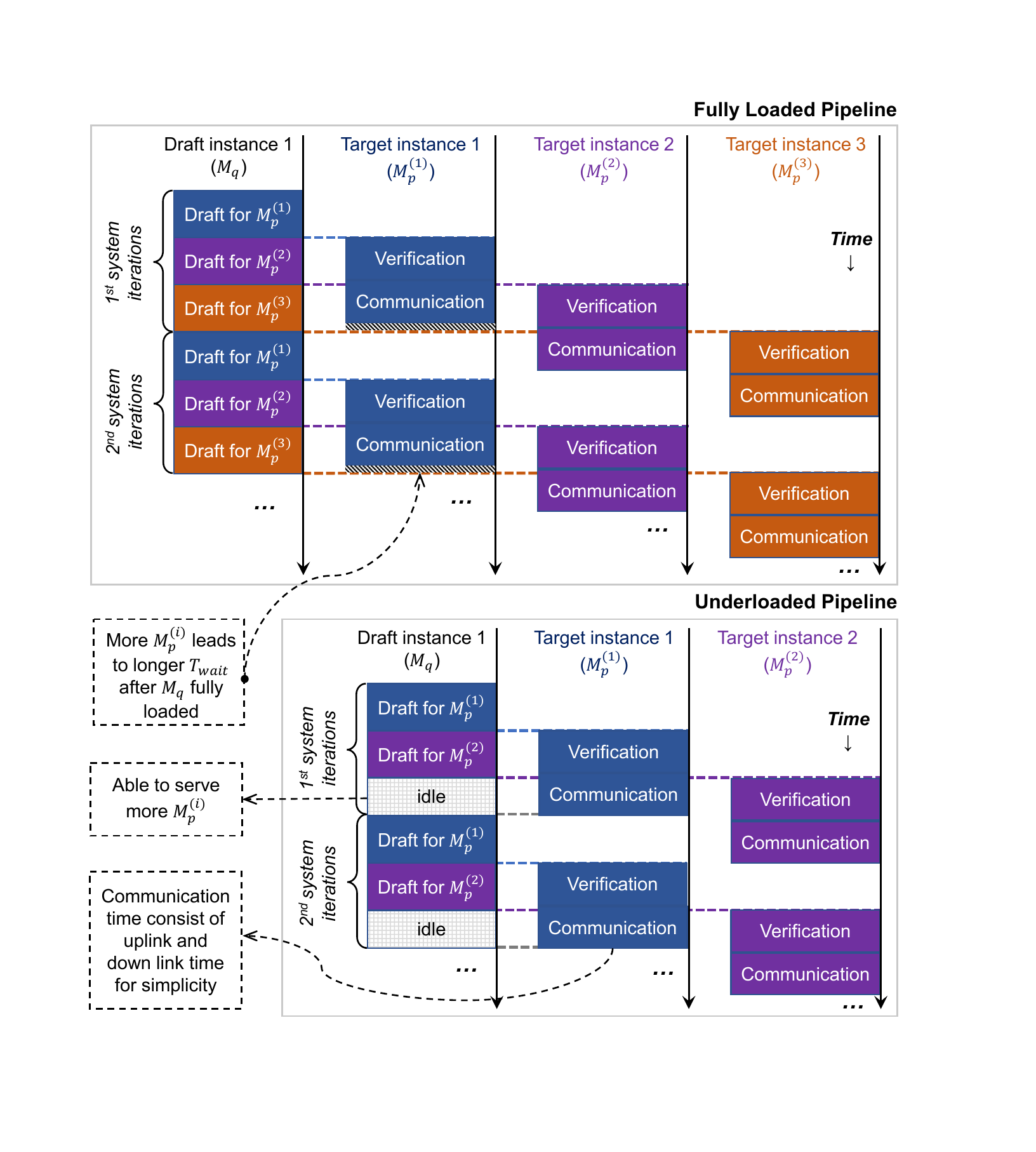}
    \caption{Execution cases of StarSD under different load conditions. 
    \textbf{Fully-loaded case} where verification and communication delays of each $M_p^{(i)}$ are overlapped, keeping $M_q$ work-conserving with no idle gaps. 
    \textbf{Under-loaded case} where the draft server idles while waiting for a target to return the next request.}
    \label{fig:system pipeline fully loaded}
\end{figure}
When \eqref{eq:busy} holds, the return time $Z(d)$ of any active $M_p^{(i)}$ can be covered by the $M_q$-side service spent on the other $N-1$ active $M_p^{(i)}$.
As a result, once the request buffer is populated, $M_q$ can continuously serve requests without observable idle gaps.
In this case, the per-iteration time is
\begin{equation}
T_{\gamma_{\rm full}} = NS(d),
\label{eq:Tgamma_full}
\end{equation}
since one iteration corresponds to serving $N$ depth-$d$ requests (one per active $M_p^{(i)}$), each taking $S(d)$ time on $M_q$.

\textbf{Under-loaded case.}

When \eqref{eq:busy} does not hold, serving the other $N-1$ active $M_p^{(i)}$ for a total of $(N-1)S(d)$ time may still be insufficient for the original $M_p^{(i)}$ to complete its return within $Z(d)$.
In this case, $M_q$ can become idle due to a temporary lack of pending requests.
We quantify this additional waiting time by defining the idle gap per iteration
\begin{equation}
T_{\mathrm{idle}} \triangleq \max\!\left\{0,\; Z(d) - (N-1)S(d)\right\}.
\label{eq:Tidle}
\end{equation}
Thus, the general per-iteration time can be expressed as
\begin{equation}
T_{\gamma} = NS(d) + T_{\mathrm{idle}}.
\label{eq:Tgamma_under}
\end{equation}
Finally, we define the throughput in the $\gamma$-th iteration as
\begin{equation}
O_\gamma \triangleq \frac{\mathbb{E}[l]_\gamma}{T_\gamma}.
\label{eq:Ogamma}
\end{equation}
This definition decouples the expected progress $\mathbb{E}[l]_\gamma$ from the time cost $T_\gamma$, making clear whether throughput is limited by acceptance or by the iteration time overheads.
It also enables case-wise interpretation through \eqref{eq:Tgamma_full} and \eqref{eq:Tgamma_under}.


\section{StarSD Implication}
This section translates the analytical conditions into StarSD implication principles and examines how throughput scales as the number of target instances increases.

\subsection{StarSD Design}
\begin{algorithm}[!t]
\caption{StarSD One-for-Many Serving Runtime}
\label{alg:starsd_min}
\hspace*{0.5em}\textbf{Input:} $M_q$, $\{M_p^{(i)}\}_{i=1}^{N}$, Prefix $\hat y$.\\
\hspace*{0.5em}\textbf{Preparation:}\\
\hspace*{0.5em} $\boldsymbol{\cdot}$ Handshake: assign each target a unique\\
\hspace*{0.5em} $\boldsymbol{\cdot}$ \texttt{tag} and establish a dedicated port.\\
\hspace*{0.5em} $\boldsymbol{\cdot}$ $\mathcal{C}[\texttt{tag}]$: per-session KV cache for \texttt{tag}.\\
\hspace*{0.5em} $\boldsymbol{\cdot}$ $\mathcal{Q}_{in}$: global request buffer $(\texttt{tag}, y)$.\\
\hspace*{0.5em} $\boldsymbol{\cdot}$ $\mathcal{Q}_{out}$: outgoing queue of replies $(\texttt{tag}, \hat{y})$.\\
\begin{algorithmic}[1]
\STATE \textit{\#\# In machine hosting $M_p^{(i)}$}
\IF {\texttt{recv} $\hat{y}$ from $M_q$}
\STATE $y \gets M_p^{(i)}(\hat y)$ \hfill \textit{// refer to Algorithm~\ref{alg:tree_search} for verification details}
\STATE \texttt{send} $y$ to $M_q$ on the port identified by \texttt{tag}
\ENDIF
\STATE
\STATE \textit{\#\# In machine hosting $M_q$}
\STATE \textit{// Receiver}
\WHILE{true}
    \STATE $(\texttt{tag}, y)\gets$ \texttt{recv} from any $M_p^{(i)}$
    \STATE \texttt{push} $(\texttt{tag}, y)$ to $\mathcal{Q}_{in}$
\ENDWHILE
\STATE \textit{// Draft inference}
\IF{$\mathcal{Q}_{in}$ is non-empty}
    \STATE $(\texttt{tag}, y)\gets$ \texttt{pop} $\mathcal{Q}_{in}$
    \STATE $(\hat{y}, \mathcal{C}[\texttt{tag}]) \gets M_q(y;\ \mathcal{C}[\texttt{tag}])$\hfill \textit{// refer to Algorithm~\ref{alg:tree_search} for inference details}
    \STATE \texttt{push} $(\texttt{tag}, \hat{y})$ to $\mathcal{Q}_{out}$
\ENDIF
\STATE \textit{// Sender}
\IF{\textbf{$\mathcal{Q}_{out}$ not empty}}
    \STATE $(\texttt{tag}, \hat{y})\gets$ \texttt{pop} $\mathcal{Q}_{out}$
    \STATE \texttt{send} $\hat{y}$ back to $M_p^{(i)}$ on the port identified by \texttt{tag}
\ENDIF
\end{algorithmic}
\end{algorithm}

We design StarSD by considering the following runtime properties.

\paragraph{Work-conserving draft execution.}
To match \eqref{eq:Tgamma_full}-\eqref{eq:Tgamma_under} under asynchronous returns, $M_q$ maintains a global request buffer $\mathcal{Q}_{in}$, a outgoing buffer $\mathcal{Q}_{out}$, and runs a non-blocking loop that always serves any pending request, making uncovered return time appear as $T_{\mathrm{idle}}$ in \eqref{eq:Tidle}. 

\paragraph{Per-target state isolation.}
To preserve the correctness of \eqref{eq:beta_cond} and \eqref{eq:tail_sum_li} under concurrency, StarSD isolates all per-target state, including KV cache and metadata, across target instances by \texttt{tag} via $\mathcal{C}[\texttt{tag}]$.

\paragraph{Independent return paths.}
To preserve the return-time decomposition $Z(d)=t_c+t_v$ in \eqref{eq:Zd}, we assign each target a unique \texttt{tag} and a dedicated port during a one-time handshake, and reuse these ports for steady-state payload transfer in Algorithm~\ref{alg:starsd_min}.

Algorithm~\ref{alg:starsd_min} provides one concrete realization of these properties, and Fig.~\ref{fig:system model} summarizes the resulting execution pipeline. Verified prefixes from multiple target instances are collected into a shared request buffer and served by the draft instance in a work-conserving manner. Drafted tokens are routed back to their corresponding targets for parallel verification, and the verified prefixes re-enter the pipeline for subsequent iterations. By decoupling request arrival, draft execution, and verification returns, this pipeline avoids cross-target head-of-line blocking and enables continuous draft-side execution under concurrent service.
\begin{figure}[t]
    \centering
    \includegraphics[width=0.48\textwidth]{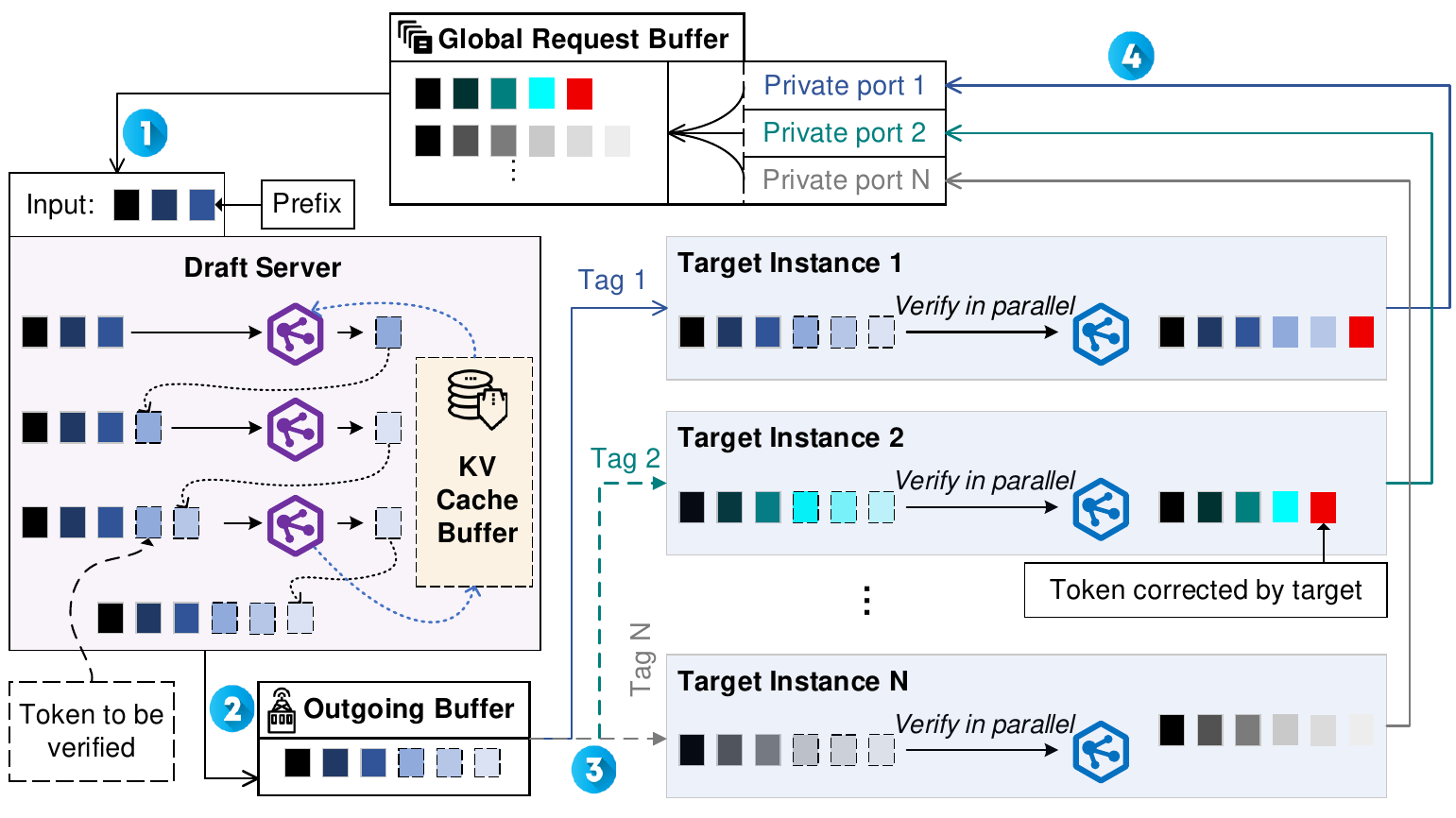}
    \caption{StarSD Pipeline. Use the global request buffer and outgoing buffer to guarantee the work continuity of the draft model.}
    \label{fig:system model}
\end{figure}

\subsection{Insights}
\label{sec:insights}
\begin{figure}[t]
    \centering
    \includegraphics[width=0.9\linewidth]{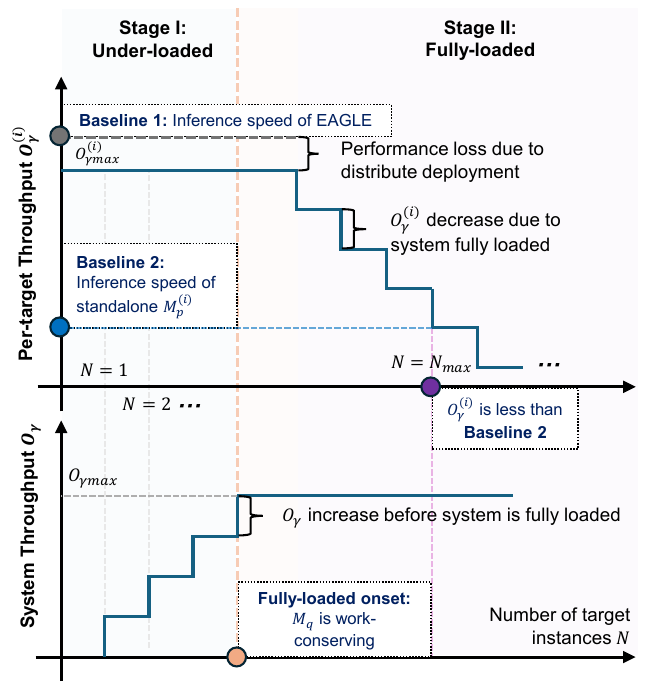}
    \caption{StarSD framework behavior when scaling the number of target instances. Baseline 1 stands for the throughput of SOTA speculative decoding that is deployed on a single accelerator. Baseline 2 is the standalone inference speed of an LLM.}
    \label{fig:throughput analysis}
\end{figure}
While StarSD enables one-for-many serving, its throughput exhibits a trade-off. Figure~\ref{fig:throughput analysis} shows a two-stage behavior as the number of active target instances grows.

\paragraph{Stage I: Under-loaded.}
\label{under load para}
When $N$ is small, $M_q$ is under-loaded, the global request buffer is frequently empty, so $M_q$ exhibits idle gaps, and adding more active \(M_p^{(i)}\) primarily improves utilization.
In this regime, the per-\(M_p^{(i)}\) throughput \(O_\gamma^{(i)}\) stays approximately constant, while the system throughput \(O_\gamma\) increases with \(N\) because the uncovered return time is gradually amortized across more concurrent requests.
This corresponds to the violation of \eqref{eq:busy}, where the uncovered return time appears as \(T_{\mathrm{idle}}\) in \eqref{eq:Tidle}, and thus \(T_\gamma\) follows \eqref{eq:Tgamma_under}.

\paragraph{Stage II: Fully-loaded.}
\label{full load para}
As \(N\) increases, the idle gap vanishes, the global request buffer becomes persistently non-empty, and \(M_q\) turns work-conserving.
Beyond this point, increasing \(N\) no longer increases the service capacity of \(M_q\) but introduces queueing in the global request buffer.
Consequently, the waiting time before a request is served increases \(T_\gamma\), which reduces the per-\(M_p^{(i)}\) throughput \(O_\gamma^{(i)}\), while \(O_\gamma\) saturates near its maximum.
Formally, the full-load onset is the smallest \(N\) such that \eqref{eq:busy} holds.
Equivalently, for
\begin{equation}
N \ge N_{\mathrm{full}} \triangleq \left\lceil \frac{Z(d)}{S(d)} \right\rceil + 1,
\end{equation}
we have \(T_{\mathrm{idle}} = 0\) and \(T_\gamma\) follows \eqref{eq:Tgamma_full}.

Importantly, admitting more $M_p^{(i)}$ can still be beneficial after full load as long as the resulting $O_\gamma^{(i)}$ remains above the throughput of a standalone $M_p^{(i)}$.
We denote by $N_{\max}$ the largest active set size such that $O_\gamma^{(i)}$ does not fall below this point. Beyond $N_{\max}$, additional concurrency degrades per-$M_p^{(i)}$ throughput below what $M_p^{(i)}$ can achieve without $M_q$.
This observation motivates a practical admission rule at $M_q$, reject new $M_p^{(i)}$ requests when $N>N_{\max}$ to preserve per-$M_p^{(i)}$ throughput higher than the standalone inference speed of $M_p^{(i)}$.

\paragraph{Packing advantage in practice.} Beyond throughput, StarSD improves practical memory packing on accelerators. Colocating one target instance with its draft model forms fixed pairs that occupy large, indivisible memory blocks. 
Under a fixed accelerator memory budget, such coarse-grained pairing limits how many target instances can be placed concurrently and leaves fragmented memory regions. 
By decoupling drafting from verification and providing the draft model as a shared service, StarSD's finer-grained placement allows accelerator memory to be primarily filled with target instances, increasing target density and enabling either more concurrent replicas or longer context lengths under the same hardware budget.

\section{Experiments}
We evaluate StarSD along three dimensions, i.e., its distributed execution characteristics, its scalability under one-for-many concurrency, and its practical draft-side performance behavior under different runtime regimes.

\subsection{Distributed Behavior of StarSD}

\begin{figure}[t]
    \centering
    \begin{minipage}[t]{0.6\columnwidth}
        \centering
        \includegraphics[width=\linewidth]{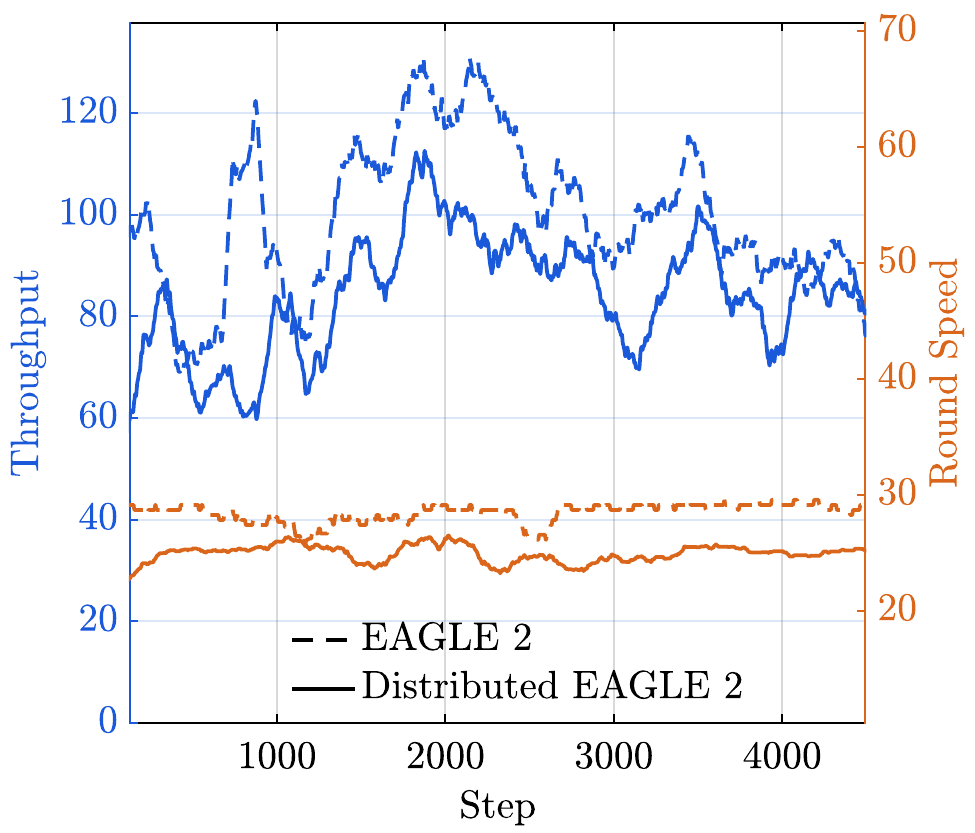}
        {\small (a) Performance Comparison.}
    \end{minipage}\hfill%
    \begin{minipage}[t]{0.39\columnwidth}
        \centering
        \includegraphics[width=\linewidth]{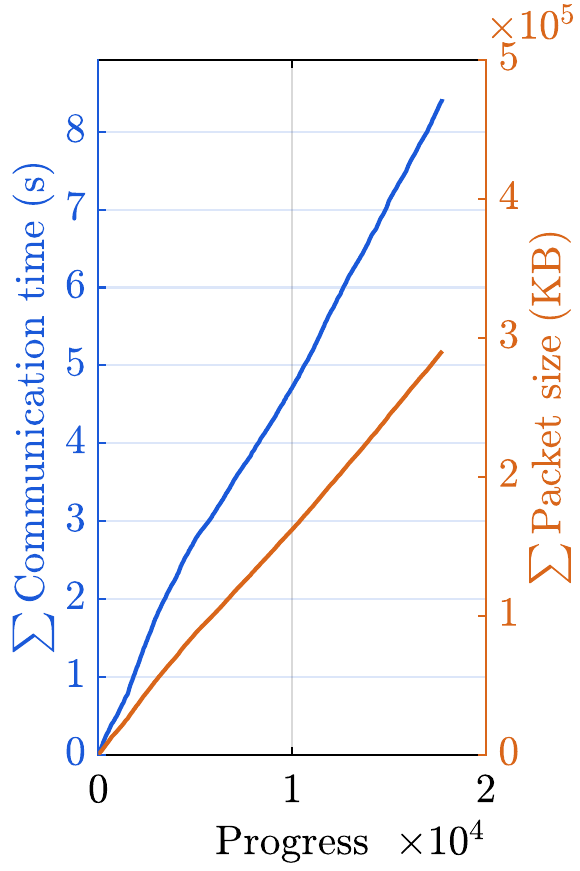}
        {\small (b) Communication Cost.}
    \end{minipage}
    \caption{Throughput loss due to distribute deployment of EAGLE and communication overhead.}
    \label{fig:distribute eagle communicate}
\end{figure}

Since StarSD does not modify model weights or the speculative decoding algorithm, performance differences arise solely from system effects, namely communication overhead introduced by distribution and potential gains from parallel execution. Thus, we evaluate StarSD on the full MT-Bench benchmark~\cite{zheng2023judging} using EAGLE~2~\cite{li2024eagle2} as a representative backend. 
To isolate the cost of distribution, we consider the $N=1$ case, where the draft instance $M_q$ and the target instance $M_p$ are placed on different accelerators. In this configuration, StarSD reduces to a distributed deployment of EAGLE~2, enabling a direct measurement of cross-device communication overhead.

As shown in Fig.~\ref{fig:distribute eagle communicate}(\textit{a}), we report throughput from two perspectives,
(1) system throughput $O_\gamma$, which fluctuates noticeably across MT-Bench due to prompt-dependent acceptance lengths, and
(2) the round speed $R_q$ of $M_q$ which is proportional to $O_\gamma$ and prove in~\ref{sec:rq_rp_metric}, it provides a smoother perspective to observe performance.
To reflect the cost of distributed deployment, we measure end-to-end serving performance in the deployed setting, so the experiment results naturally include cross-device overheads such as communication latency shown in Fig.~\ref{fig:distribute eagle communicate}(\textit{b}), message serialization and deserialization, and GPU-CPU memory transfers. Our implementation is deployment-agnostic and supports multi-accelerator execution either within a single machine or across machines. Accordingly, we use socket-based communication rather than hardware-specific backends, e.g., NVLink, trading some performance for portability. This observation aligns with recent findings in distributed LLM serving that cross-accelerator data movement can incur non-trivial overhead and lead to a measurable performance drop under equivalent accelerator resources~\cite{2026arXiv260108833L}.

\subsection{Scaling Behavior of StarSD}
\begin{figure}[t]
    \centering
    \includegraphics[width=\linewidth]{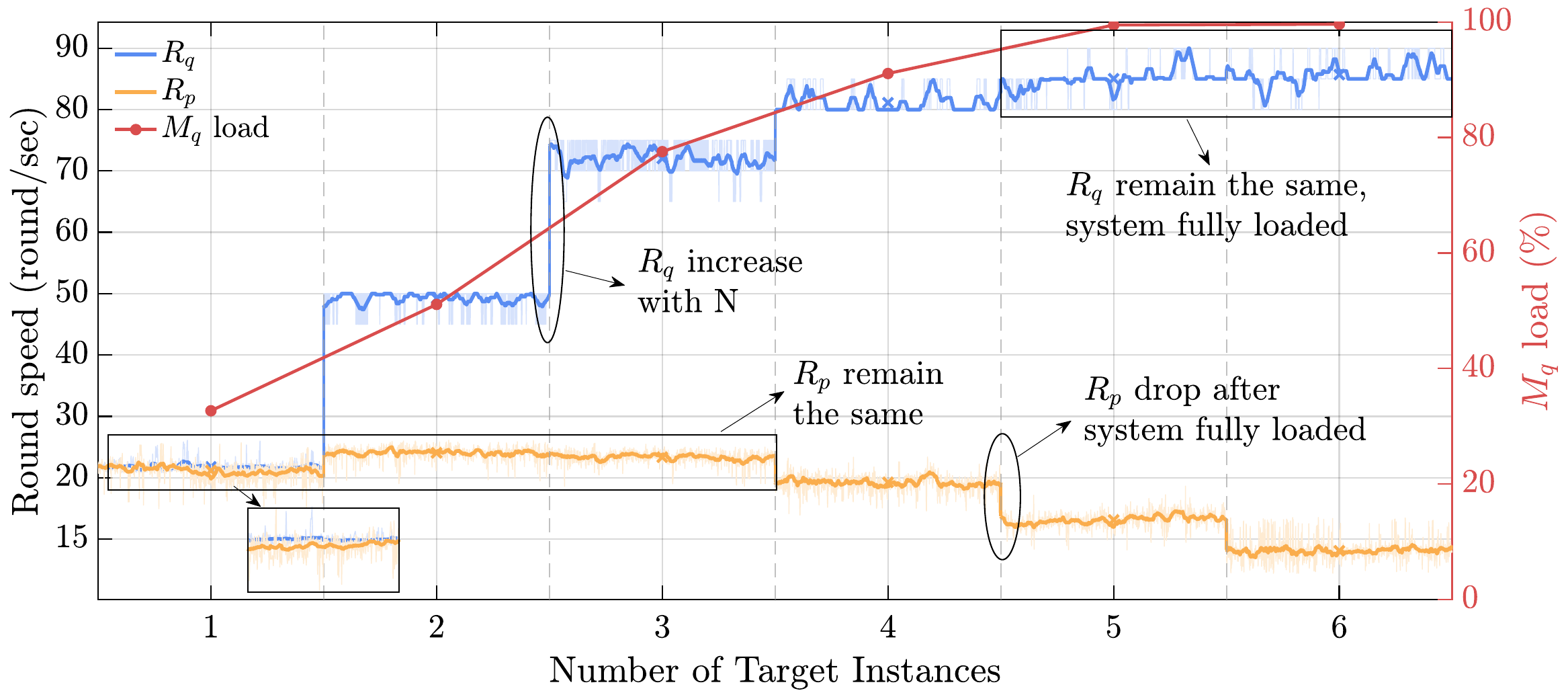} 
    \caption{Round speed $R_q$ and $R_p$, load of $M_q$ under varying numbers of target instances $N$.}
    \label{fig:1_for_N_throughput analysis}
\end{figure}

\begin{table}[t]
\centering
\caption{System performance metrics as the number of target instances $N$ increases.}
\label{tab:system_metrics_T}
\setlength{\tabcolsep}{10pt}
\renewcommand{\arraystretch}{1.15}
\resizebox{\columnwidth}{!}{%
\begin{tabular}{l|c|cccccc}
\toprule
\textbf{Metric} & \textbf{EAGLE} & \multicolumn{6}{c}{\textbf{StarSD with different numbers of target instances $N$}} \\
\cmidrule(lr){3-8}
 &  & 1 & 2 & 3 & 4 & 5 & 6 \\
\midrule

\multicolumn{8}{l}{\textit{{\textbf{Throughput (tokens/s)}}}} \\
$O_\gamma$        & 94.3  & 79.9  & 167.6 & 246.3 & 266.6 & 285.5 & \textbf{288.7} \\
$O_\gamma^{(i)}$  & 118.5 & 86.6  & \textbf{100.9} & 98.3 & 80.94 & 69.49. & 57.78 \\
\cmidrule(lr){1-8}

\multicolumn{8}{l}{\textit{{\textbf{Time breakdown (ms)}}}} \\
$t_s$            & 1.78 & 2.24 & 1.36 & 1.36 & 1.42 & 1.46 & 1.46 \\
$T_{\text{wait}}$ & \textemdash & 0.01 & 0.50 & 0.60 & 3.90 & 13.0 & 22.3 \\
$T_{\text{idle}}$ & \textemdash & 30.6 & 9.8  & 3.4  & 0.6  & 0.0  & 0.0 \\
\cmidrule(lr){1-8}
\multicolumn{8}{l}{\textit{{\textbf{Round speed (round/s)}}}} \\
$R_q$             & 26.8 & 20.90 & 49.25 & 72.26 & 81.54 & 85.52 & \textbf{85.84} \\
$R_p$             & 26.8 & 20.90 & \textbf{24.07} & 23.38 & 19.73 & 16.64 & 14.10 \\
\cmidrule(lr){1-8}

\multicolumn{8}{l}{\textit{{\textbf{Utilization (\%)}}}} \\
$M_q$ load   & \textemdash & 33.3 & 50.3 & 75.4 & 91.4 & \textbf{99.5} & \textbf{99.6} \\
\bottomrule
\end{tabular}%
}
\end{table}




To evaluate one-for-many scaling, we vary the number of concurrent target instances $N$ served by a single draft instance $M_q$ and examine how system performance evolves with increasing concurrency. Fig.~\ref{fig:1_for_N_throughput analysis} reports the round speeds $R_q$ and $R_p$ under different $N$. 
As $N$ increases from 1 to 3, $R_q$ grows substantially, while $R_p$ changes only modestly. This behavior matches the \emph{under-loaded} regime in Section~\ref{under load para}, where serving additional targets reduces draft-side idle gaps by overlapping verification and communication latency across requests. 
At $N=4$, $R_q$ approaches its maximum, indicating that $M_q$ is close to full utilization and the system begins to enter the \emph{fully-loaded} regime described in Section~\ref{full load para}. 
Further increasing $N$ provides only marginal improvement in $R_q$, while $R_p$ decreases noticeably due to queueing delays before requests are processed by $M_q$.

To demonstrate the universality of StarSD, we evaluate both intra-node and inter-node configurations as detailed in Appendix~\ref{app:deployment}, and summarize system performance for up to six target instances in Table~\ref{tab:system_metrics_T}.
In the intra-node setting, multiple target instances run on distinct GPUs within a single server, while in the inter-node setting, target instances are distributed across multiple servers connected by a local-area network. Specifically, our deployment uses seven GPUs across two servers. One server hosts four RTX~4090 GPUs, running a draft model $M_q$ and three target instances, while the other server hosts three A100 GPUs, each running one target instance. As $N$ increases, targets are first activated on the RTX~4090 server, and targets on the A100 server are used when $N>3$. The shared $M_q$ server serves all active targets across both machines.
All metrics are measured end-to-end and therefore include communication latency and cross-machine serialization overhead, showing that StarSD remains portable and effective across both single-server and multi-server deployments. Beyond throughput and round speed, we report two delay metrics, i.e., the draft-side idle time $T_{\text{idle}}$, which captures how long $M_q$ remains idle between requests, and the queueing wait time $T_{\text{wait}}$, which measures how long a request waits in the global buffer before service. These metrics explain a counter-intuitive trend in Table~\ref{tab:system_metrics_T}, i.e., when $N=1$, $R_p$ is lower than at $N=2$, consistent with a large $T_{\text{idle}}$. Empirically, excessive idle gaps push $M_q$ into an unfavorable operating regime and degrade its per-token inference speed.

\subsection{Execution Regimes and Performance Dynamics}
\begin{figure*}[t]
    \centering
    \begin{minipage}[t]{0.32\textwidth}
        \centering
        \includegraphics[width=\linewidth]{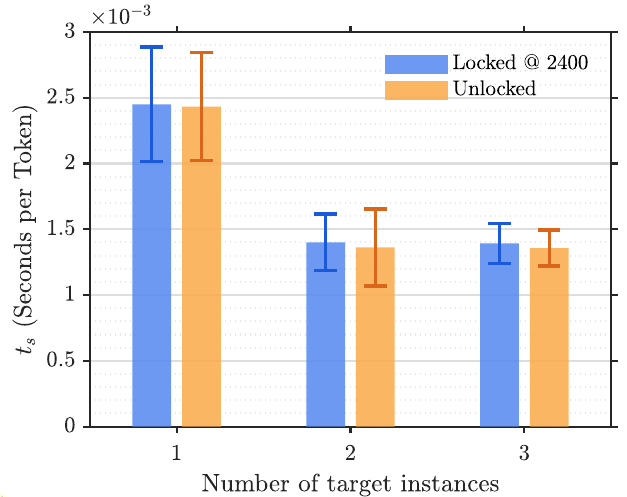}
        {\small (a) Effect of locking GPU frequency.}
    \end{minipage}
    \hfill
    \begin{minipage}[t]{0.31\textwidth}
        \centering
        \includegraphics[width=\linewidth]{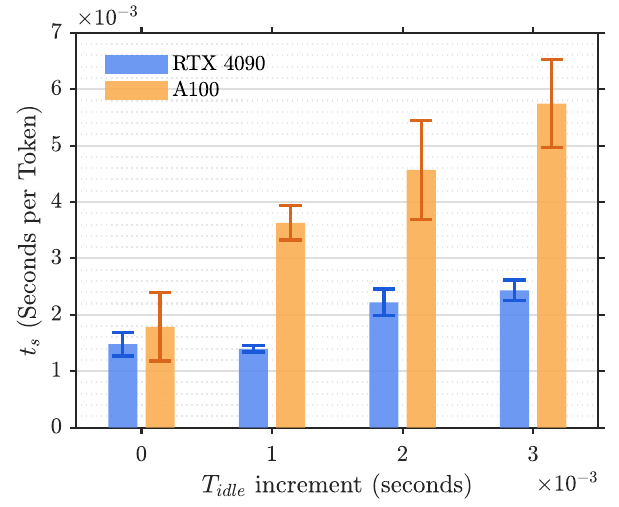}
        {\small (b) Effect of $T_{\text{idle}}$.}
    \end{minipage}
    \hfill
    \begin{minipage}[t]{0.31\textwidth}
        \centering
        \includegraphics[width=\linewidth]{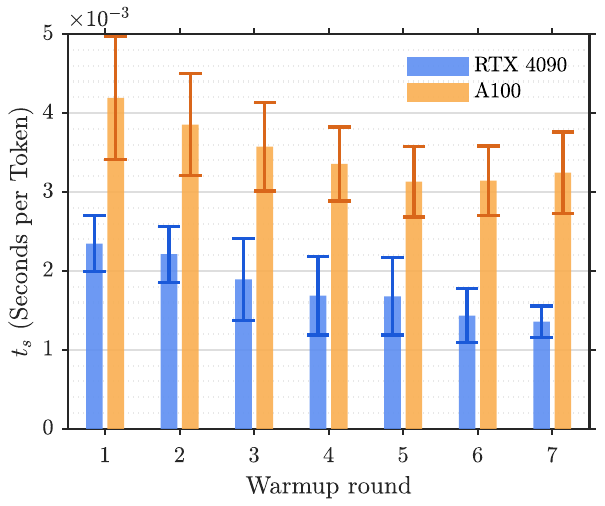}
        {\small (c) Effect of warm-up rounds.}
    \end{minipage}
    \caption{Per-token generation time analysis of $M_q$ under warm-up at $T_{\text{idle}}=0.03\,$s, GPU frequency control, and $T_{\text{idle}}$ variation.}
    \label{fig:three_ablation_panels}
\end{figure*}

As suggested by the results in Table~\ref{tab:system_metrics_T}, a major source of StarSD's system-level performance gain comes from eliminating cold-start and idle gaps on the draft side. To further analyze this effect and verify that the observed variation in the inference speed of $M_q$ is not specific to a particular accelerator, we conduct controlled experiments by forming a single draft–target pair and evaluating it independently on an RTX~4090 GPU and on an A100 GPU.

We first rule out GPU frequency scaling as a confounding factor by instrumenting the graphics clock before each inference. As shown in Fig.~\ref{fig:three_ablation_panels}(\textit{a}), the same slowdown trend appears under both the default Dynamic Voltage and Frequency Scaling (DVFS) policy and a fixed clock of 2400 MHz, indicating that frequency changes are not the cause. 
To precisely control idle behavior, we then inject artificial delays on the draft side by pausing $M_q$ for a prescribed duration after receiving each request, thereby inducing a target $T_{\text{idle}}$. Figure~\ref{fig:three_ablation_panels}(\textit{b}) shows that per-token inference time remains stable when $T_{\text{idle}}$ is either small or large, but degrades in an intermediate regime where increasing $T_{\text{idle}}$ leads to a monotonic slowdown and higher run-to-run variance. Together, these results indicate that sustained, continuous execution keeps the draft model in a favorable operating state, while intermittent idle gaps introduce transient inefficiencies that increase both latency and variability.

We further examine the warm-up effect shown in Fig.~\ref{fig:three_ablation_panels}(\textit{c}) by inserting additional draft-side executions after an idle gap and before normal generation begins. 
In this setup, $M_q$ generates a small number of warm-up tokens to transition back toward a steady execution state. 
As the number of warm-up rounds increases, the per-token inference speed improves and then stabilizes, with diminishing gains beyond a modest threshold. 
At low warm-up counts, the execution behavior is comparable to the $N=1$ setting in Table~\ref{tab:system_metrics_T}, which exhibits an idle gap of approximately $T_{\text{idle}}\approx 0.03$ second and degraded performance. 
As warm-up rounds increase, the effective idle impact is reduced to about $0.014$ second, leading to faster and more stable inference.

\section{Related Work}
\paragraph{Speculative Decoding.}
Speculative decoding accelerates autoregressive generation by having a lightweight draft propose tokens and a stronger target verify them while preserving the target distribution~\cite{pmlr-v202-leviathan23a,chen2023specsampling}.
Follow-up work improves efficiency by proposing and verifying multiple candidates per iteration.
Tree-based speculative inference verifies a token tree in parallel to better exploit GPU parallelism~\cite{miao2024specinfer, NEURIPS2024_ea1f5f08, 10.1162/tacl_a_00735}, and other variants expand proposals via multi-branch or multi-head drafting under exact acceptance rules~\cite{cai2024medusa}.

\paragraph{LLM Serving Under Memory Pressure.}
LLM serving is often bottlenecked by the KV-cache footprint, which limits batch size and makes co-locating an additional draft model challenging.
vLLM reduces KV-cache fragmentation via paging-style management~\cite{kwon2023efficient}.
Beyond single-GPU settings, systems disaggregate prefill and decode to reduce interference and improve goodput~\cite{zhong2024distserve}, and phase-splitting designs similarly separate prompt and token phases across machines~\cite{patel2024splitwise}.
These results suggest that, under realistic batch sizes and KV-cache pressure, hosting both a large target and a draft on one accelerator can be impractical.

\paragraph{Distributed and Networked Speculative Decoding.}
Placing draft and target on different accelerators turns speculative decoding into a distributed pipeline where communication and asynchronous returns shape utilization.
Prior work explores edge settings, e.g., lightweight edge drafts with server-side verification~\cite{li2025sled}, and distributed speculative inference that leverages multiple drafters, targets and task parallelism~\cite{timor2024distributed,2025arXiv250521594V,2025arXiv250310325G}.
SpecEdge further splits inference between edge and server GPUs via speculative decoding and pipeline-aware scheduling, exchanging only token outputs over the network~\cite{park2025specedge}.
Overall, once drafting and verification are disaggregated, performance depends not only on acceptance behavior but also on communication efficiency and scheduling.

\section{Conclusion}
StarSD is a one-for-many speculative decoding framework that allows a single draft model to efficiently serve multiple target instances. 
By keeping the draft side work-conserving under asynchronous verification, StarSD eliminates idle gaps that limit the efficiency in distributed deployments. 
Experiments across heterogeneous settings show that this design improves utilization and throughput without modifying model weights or generation semantics, making speculative decoding more practical under realistic memory and deployment constraints.

\section*{Impact Statement}
StarSD enables one-for-many speculative decoding by decoupling drafting and verification across distributed accelerators, allowing a single draft model to serve multiple target instances in a work-conserving manner. This can improve utilization and reduce deployment overheads under memory and resource constraints, supporting more flexible allocation across heterogeneous devices in both cluster and edge serving environments.
However, StarSD may also introduce risks. By lowering the cost of LLM inference, it can make large-scale misuse easier. More importantly, the shared draft model can become a concentrated trust point, it may transiently hold multiple sessions' prompts, intermediate prefixes, and cached decoding state across served targets. If the draft service is compromised, an attacker could potentially access sensitive user inputs or session traces, leading to severe privacy leakage.
To mitigate these risks, deployments should treat the draft service as a high-sensitivity component. Enforce strong isolation between sessions, minimize and protect in-memory and on-disk artifacts including KV caches and logs, use authenticated and encrypted channels for inter-node communication, and apply access control, rate limiting, and auditing to detect abuse and intrusion.


\bibliography{Ref}
\bibliographystyle{icml2026}

\newpage
\appendix
\onecolumn


\appendix
\numberwithin{equation}{section}

\section{Overview of Supplementary Material}
\label{app:overview}

This appendix provides supplementary algorithmic descriptions, deployment illustrations, and experimental artifacts that support the results in the main text.
For reproducibility, command-line snippets and representative log excerpts are highlighted in \texttt{typewriter font}.
The appendix is organized as follows:

\begin{enumerate}
    \item \textbf{Motivation of StarSD:} We motivate StarSD from practical deployment constraints and the resulting inefficiencies of distributed speculative decoding, explaining why a one-for-many draft service is desirable in~\Cref{motivation}.
    
    \item \textbf{Mathematical Justification:}
    We provide further mathematical Justification in~\Cref{app:math}.
    
    \item \textbf{Algorithm Details:}
    We describe the EAGLE-based speculative decoding procedure used in StarSD, including the key steps of proposal, verification, and correction, in~\Cref{app:eagle}.

    \item \textbf{Implementation Details:}
    We present system-level implementation details of StarSD, including the one-time connection establishment protocol, the steady-state request execution loop, the decoupled communication that computes runtime at the draft server, and tag-indexed per-session state management, e.g., KV cache lifecycle and speculative metadata in~\Cref{app:impl}.

    \item \textbf{Extended GPU Experiment Details:}
    We provide additional reproducibility artifacts for the GPU frequency-control study and draft-side inference-time measurements in~\Cref{app:gpu_commands,app:draft_inference}.

    \item \textbf{Deployment Illustrations:}
    We summarize representative deployment topologies and illustrate the end-to-end distributed pipeline in~\Cref{app:deployment}.

    \item \textbf{Discussion and Future Directions:}
    We discuss limitations of the current design and outline future directions in~\Cref{app:future}.
\end{enumerate}

\section{Motivation}
\label{motivation}
Recent advances in speculative decoding have largely focused on improving the draft model by better aligning its token distribution with the target model. However, most existing frameworks implicitly assume that the draft instance $M_q$ and the target instance $M_p$ can be co-located on a single accelerator, e.g., a GPU. In practical deployments, this assumption often does not hold. Many real systems operate under heterogeneous and resource-constrained settings, including edge clusters and mid-scale cloud servers, where accelerator memory is limited and shared by multiple services. The weights may barely fit, but the remaining memory must still accommodate KV cache and runtime buffers. In practice, this sharply reduces the KV cache budget, hence affect the context length, which is especially problematic in multi-tenant deployments where one GPU also serves other inference tasks. Even with compression or quantization, co-locating $M_q$ and $M_p$ can remain impractical or introduce nontrivial accuracy, engineering, and operational costs.

A natural response is to disaggregate the deployed instances $M_q$ and $M_p$ across different accelerators, e.g., GPU. Distributed speculative decoding has been explored in prior work, but most existing designs do not address a core inefficiency inherent to the draft-then-verify workflow. When the target instance $M_p$ is verifying draft tokens, the accelerator hosting the draft instance $M_q$ is idle; when $M_q$ is drafting, the $M_p$ accelerator is underutilized. This alternating execution pattern reduces effective hardware utilization and makes $M_q$ execution bursty, with short compute phases separated by idle gaps induced by verification latency and inter-accelerator communication. These idle gaps have system-level consequences that are also largely overlooked in prior work. Even when the decoding depth is unchanged, bursty execution prevents $M_q$ from sustaining a stable high-throughput execution state and may trigger repeated per-iteration transients (e.g., kernel reactivation and cache coldness), leading to increased draft-step latency and higher variance across decoding iterations.

More fundamentally, the generated output is ultimately dictated by the target instance $M_p$, while $M_q$ serves as auxiliary compute overhead. From a system perspective, this overhead should be amortized across multiple target instances rather than dedicated to a single $M_p$. These observations motivate a one-for-many design in which a single $M_q$ multiplexes requests from multiple $M_p^{(i)}$. Such a design allows verification responses to arrive more continuously, reduces idle gaps on the draft accelerator, and improves system-level utilization. 

\section{Mathematical Justification}
\label{app:math}




\subsection{Metric Justification of Draft and Verify Round Speed}
\label{sec:rq_rp_metric}

We consider a fixed dataset $\mathcal{D}$, a fixed decoding configuration, and a fixed model pair $(M_p,M_q)$ with depth $d$.
Across deployments, we keep the model weights and the decoding algorithm unchanged.
Hence, for any prefix $y$ encountered during decoding on $\mathcal{D}$, the next-token distributions $p(\cdot\mid y)$ and $q(\cdot\mid y)$ are invariant up to numerical nondeterminism, and the induced acceptance behavior is invariant in expectation.
Under the standard approximation, the expected accepted length per draft proposal
\begin{equation}
\bar{\ell}\triangleq \mathbb{E}[l(d)],
\label{eq:ell_bar}
\end{equation}
is deployment-invariant.

\paragraph{Round speeds.}
Let $R_q$ denote the rate per unit time at which the draft model executes the token-generation routine that produces a depth-$d$ proposal, each such execution yields an accepted length $l(d)$ in the subsequent verification step.
Let $R_p$ denote the rate per unit time at which a target instance executes a verification routine.
Both $R_q$ and $R_p$ are measured in the steady state.

\paragraph{Throughput as tokens per unit time.}
Recall that the system throughput~\eqref{eq:Ogamma} from the draft-side view is
\begin{equation}
O_\gamma \triangleq \frac{\mathbb{E}[l]_\gamma}{T_\gamma},
\label{eq:Ogamma_recall}
\end{equation}
where $\mathbb{E}[l]_\gamma$ is the expected number of accepted tokens contributed in an interval of duration $T_\gamma$.
From the target-side view, we similarly define the per-target throughput for instance $i$ as
\begin{equation}
O_\gamma^{(i)} \triangleq \frac{\mathbb{E}[l]_\gamma^{(i)}}{T_\gamma}.
\label{eq:Ogamma_i_recall}
\end{equation}

\paragraph{Positive correlation with $R_q$.}
Over any time window of length $T_\gamma$, the number of draft-generation executions is approximately $R_q T_\gamma$ in expectation.
Each execution contributes an accepted length with expectation $\bar{\ell}$.
Therefore,
\begin{equation}
\mathbb{E}[l]_\gamma \;\approx\; (R_q T_\gamma)\,\bar{\ell},
\qquad
O_\gamma \;=\;\frac{\mathbb{E}[l]_\gamma}{T_\gamma}\;\approx\;\bar{\ell}\,R_q.
\label{eq:Ogamma_vs_Rq}
\end{equation}
Since $\bar{\ell}>0$ under fixed $\mathcal{D}$ and decoding setup, $O_\gamma$ is monotone increasing in $R_q$, i.e., $R_q$ is positively correlated with $O_\gamma$.

\paragraph{Positive correlation with $R_p$.}
Fix a target instance $M_p^{(i)}$.
Over a window of length $T_\gamma$, the expected number of verification executions by this instance is approximately $R_p T_\gamma$.
Each verification round finalizes one proposal and advances the decoded sequence by an accepted length with expectation $\bar{\ell}$.
Thus,
\begin{equation}
\mathbb{E}[l]_\gamma^{(i)} \;\approx\; (R_p T_\gamma)\,\bar{\ell},
\qquad
O_\gamma^{(i)} \;=\;\frac{\mathbb{E}[l]_\gamma^{(i)}}{T_\gamma}\;\approx\;\bar{\ell}\,R_p,
\label{eq:Ogamma_i_vs_Rp}
\end{equation}
which implies that $O_\gamma^{(i)}$ is monotone increasing in $R_p$ and hence $R_p$ is positively correlated with the throughput of the $i$-th target instance.

\paragraph{Bottleneck regime.}
In practice, draft and verification must be matched, so the realized throughput is governed by the slower side.
This is consistent with the upper-bound intuition
\begin{equation}
O_\gamma \;\lesssim\; \bar{\ell}\,\min\{R_q,\,R_p\},
\label{eq:Ogamma_min_bound}
\end{equation}
and increasing either $R_q$ or $R_p$ cannot decrease throughput.

\paragraph{Special case: 1-for-1 in StarSD.}
When there is only one target instance, i.e., $N=1$, the system has a single verification stream, so the per-target throughput coincides with the system throughput
\begin{equation}
O_\gamma^{(i)} = O_\gamma.
\label{eq:Ogamma_eq_Oi_N1}
\end{equation}
Moreover, the execution forms a tightly coupled loop, where
each draft-generation round on $M_q$ produces exactly one proposal, and each proposal triggers exactly one verification round on $M_p$.
Therefore, in steady state, the number of draft rounds and verification rounds are equal over any long horizon, implying
\begin{equation}
R_q = R_p.
\label{eq:Rp_eq_Rq_N1}
\end{equation}
Combining~\eqref{eq:Ogamma_i_vs_Rp} and~\eqref{eq:Ogamma_eq_Oi_N1}, both metrics induce the same monotone proxy for throughput in the 1-for-1 setting:
\begin{equation}
O_\gamma = \bar{\ell}\,R_q = \bar{\ell}\,R_p.
\label{eq:Ogamma_N1}
\end{equation}

\section{EAGLE Algorithm Details}
\label{app:eagle}

\paragraph{Setting and notation.}
Let $y$ denote the \emph{verified prefix} at the beginning of one speculative iteration.
EAGLE uses a draft model $M_q$ to construct a $k$-ary proposal tree of depth $d$, and a target model $M_p$ to verify candidate paths.
A candidate path is represented as $\texttt{tree\_path} = [x_1,\ldots,x_d]$.
During tree expansion, the draft model $M_q$ produces next-token distributions at each depth. We denote the distribution at depth $t$ by $q_t(\cdot)$.
The correction step uses the corresponding draft probability $q_{l+1}(x)$ in
$\texttt{norm}(\max(0, p_{l+1}(x)-q_{l+1}(x)))$.
During verification, running $M_p$ along the prefixes induced by a path produces target-side distributions $p_1(\cdot),\ldots,p_{d+1}(\cdot)$.

\begin{algorithm}[!t]
\caption{EAGLE Speculative Decoding}
\label{alg:tree_search}
\begin{algorithmic}[1]
\STATE \textbf{Input:} $M_q$, $M_p$, verified prefix $y$, depth $d$, top-$k$ $k$
\STATE \textbf{Output:} $y + \texttt{target\_path}$ and updated \texttt{hidden\_state}
\vspace{0.3em}

\STATE \texttt{accept\_length} $\gets 0$ \hfill \textit{// best accepted length among all verified paths}
\STATE \texttt{target\_path} $\gets [\ ]$ \hfill \textit{// path selected for this iteration}
\STATE

\IF{\texttt{First Iteration}}
    \STATE \textit{// Initialize root token and hidden state to start $M_q$}
    \STATE $x,\ \texttt{hidden\_state} \sim p(y)$
    \STATE \texttt{tree} $\gets x$ \hfill \textit{// root of the proposal tree}
\ELSE
    \STATE \textit{// Use the last verified token as the tree root}
    \STATE \texttt{tree} $\gets y[-1]$
\ENDIF
\STATE

\STATE \textit{// Build a depth-$d$ proposal tree with branching factor $k$}
\FOR{$t = 1$ to $d$}
    \STATE \textit{// KV-cache maintenance for draft expansion (prefix init or incremental update)}
    \IF{\texttt{kv\_cache} is empty}
        \STATE \texttt{kv\_cache} $\gets$ \textsc{InitKV}$(y,\ \texttt{hidden\_state})$
        \hfill \textit{// cache corresponds to the verified prefix}
    \ELSE
        \STATE \texttt{kv\_cache} $\gets$ \textsc{UpdateKV}$(\texttt{target\_path},\ \texttt{kv\_cache},\ \texttt{hidden\_state})$
        \hfill \textit{// follow the currently selected path}
    \ENDIF

    \STATE \textit{// Expand one tree layer in parallel using $M_q$ (top-$k$ per node)}
    \STATE \texttt{tree} $\gets$ \textsc{ExpandLayer}$(\texttt{tree},\ \texttt{hidden\_state},\ M_q,\ k)$
\ENDFOR
\STATE

\STATE \textit{// Verify candidate paths in parallel}
\FOR{$n = 1$ to $k^{d}$}
    \STATE \texttt{tree\_path} $\gets$ \texttt{tree}[$n$]
    \STATE \textit{// \texttt{tree\_path} corresponds to $\pi=[x_1,\ldots,x_d]$}
    \STATE \textit{// Run $M_p$ along prefixes induced by the path, can be executed in parallel}
    \STATE $p_1(x),\ldots,p_{d+1}(x) \gets M_p(y),\ldots,M_p(y+[x_1,\ldots,x_d])$

    \STATE \textit{// Determine accept length on this path via stochastic test}
    \STATE $r_1 \sim U(0,1),\ldots,r_d \sim U(0,1)$
    \STATE $l \gets \min\left(\{i-1 \mid 1 \le i \le d,\ r_i > p_i(x)\} \cup \{d\}\right)$

    \STATE \textit{// Prepare continuation (either accept full prefix, or apply one-step correction)}
    \STATE $p'(x) \gets p_{l+1}(x)$
    \IF{$l < d$}
        \STATE $p'(x) \gets \texttt{norm}\!\left(\max(0,\ p_{l+1}(x) - q_{l+1}(x))\right)$
        \STATE \texttt{tree\_path} $\gets$ \texttt{tree\_path}[:$l$].\texttt{append}$(t \sim p'(x))$
        \hfill \textit{// accept first $l$ tokens, then sample one corrected token}
    \ELSE
        \STATE \texttt{tree\_path} $\gets$ \texttt{tree\_path}[:$l$]
        \hfill \textit{// accept the whole depth-$d$ path}
    \ENDIF

    \STATE \textit{// Select the path with the largest accept length}
    \IF{$l >$ \texttt{accept\_length}}
        \STATE \texttt{accept\_length} $\gets l$
        \STATE \texttt{target\_path} $\gets$ \texttt{tree\_path}
    \ENDIF
\ENDFOR
\STATE

\STATE \textbf{return} $y + \texttt{target\_path}$,\ \texttt{hidden\_state}
\end{algorithmic}
\end{algorithm}

\textbf{Initialization.}
EAGLE starts each iteration from a verified prefix $y$.
In the first iteration, $M_q$ samples a root token $x$ and an accompanying \texttt{hidden\_state} from the distribution conditioned on $y$ provided by $M_p$, in later iterations, it anchors the new proposal tree at the last verified token $y[-1]$.

\textbf{Tree expansion with caching.}
The proposal tree is expanded layer-by-layer to depth $d$ with top-$k$ branching.
Before expanding each layer, the algorithm initializes the draft KV cache from $y$ or incrementally updates it following the currently selected \texttt{target\_path}.
The expansion routine \textsc{ExpandLayer} is executed in parallel over the current tree frontier, and the draft probabilities $q_t(\cdot)$ needed by the correction step are stored in the tree metadata.

\textbf{Verification and accept-length computation.}
For each candidate path $\pi=[x_1,\ldots,x_d]$, the target model $M_p$ is evaluated on the sequence of induced prefixes
$y$, $y+x_1$, \ldots, $y+[x_1,\ldots,x_d]$ to produce $p_1(\cdot),\ldots,p_{d+1}(\cdot)$.
The algorithm samples $r_1,\ldots,r_d$ and gets the accepted length $l$ using the test in Alg.~\ref{alg:tree_search}.

\textbf{Correction and best-path selection.}
If $l<d$, the algorithm forms a corrected distribution
$p'(x)\gets \texttt{norm}(\max(0,\,p_{l+1}(x)-q_{l+1}(x)))$ and appends one corrected token after the accepted prefix.
Otherwise, it keeps the accepted prefix of length $l$ (which equals $d$ in this case).
Finally, it selects the path with the largest accepted length and returns $y+\texttt{target\_path}$ along with the updated \texttt{hidden\_state} of the selected path.
\FloatBarrier

\section{Implementation Details of StarSD}
\label{app:impl}

This appendix provides implementation-level details and covers the following topics.

\begin{itemize}
    \item One-time connection establishment protocol.
    \item Steady-state request execution loop and payload handling.
    \item Decoupling communication from draft inference on $M_q$ using separate logical components or threads.
    \item Tag-indexed state management, including KV cache lifecycle and other per-session metadata.
    \item How these details realize the design requirements in Section~5.
    \item Work-conserving execution to support \eqref{eq:Tgamma_full}-\eqref{eq:Tgamma_under}.
    \item Per-session isolation to keep $\beta_i(y)$ in \eqref{eq:beta_cond} and $l_i(d)$ in \eqref{eq:tail_sum_li} well-defined.
    \item Clean time decomposition to preserve $Z(d)=t_c+t_v$ in \eqref{eq:Zd} and $S(d)=d t_s$ in \eqref{eq:Sd}.
\end{itemize}

\subsection{Connection Establishment}
\label{app:impl:handshake}
\paragraph{One-time handshake per $M_p^{(i)}$.}
Each $M_p^{(i)}$ performs an initial handshake with $M_q$ through a public port.
Upon the first request from this $M_p^{(i)}$, $M_q$ allocates (i) a private communication port and (ii) a unique tag.
All subsequent messages between this $M_p^{(i)}$ and $M_q$ are sent through the private port, and the tag is attached to every request to identify the corresponding per-session states at $M_q$.
By isolating the control plane, i.e., handshake from the data plane, i.e., per-iteration messages, we ensure $t_c$ in \eqref{eq:Zd} corresponds to the minimal, repeated payload transfer rather than connection management overhead.

\subsection{Steady-State Request Execution}
\label{app:impl:loop}
\paragraph{Request execution loop.}
Given a user prompt, $M_p^{(i)}$ begins decoding and produces an initial token.
Then, in each iteration, $M_p^{(i)}$ sends the verified prefix or the newly verified token, along with its tag, to $M_q$ via the private port.
The payload follows the minimal-transfer design of the distributed EAGLE setting to reduce the per-iteration communication cost.
Upon receiving a tagged request, $M_q$ pushes it into the global request buffer.
When the $M_q$ inference worker becomes available, it selects a buffered request using a work-conserving scheduling policy (e.g., first-in-first-out).
Before inference, $M_q$ retrieves the corresponding KV cache using the tag, executes draft inference, and updates the KV cache upon completion.
The draft output is written to the per-$M_p^{(i)}$ communication channel and sent back to $M_p^{(i)}$ via the corresponding private port, where it is verified and used to form the next request.

\subsection{Decoupling Communication and Draft Inference on $M_q$}
\label{app:impl:decouple}
\paragraph{Separate communication and inference components.}
As shown in Fig.~\ref{fig:system model}, $M_q$ consists of two logical components: (i) a communication component that receives requests and appends them to the global request buffer, and (ii) an inference component that repeatedly selects and executes buffered requests as long as the buffer is non-empty.
Each active $M_p^{(i)}$ is associated with its own private port and a dedicated communication channel at $M_q$, preventing cross-session interference while keeping $M_q$ work-conserving.
This decoupling is required for the service-time model $S(d)=d t_s$ in \eqref{eq:Sd} in which I/O stalls do not backpressure the compute loop, so the measured draft-side service remains dominated by pure compute rather than network latency.
In addition, it helps keep the per-iteration transfer cost accounted for in $Z(d)=t_c+t_v$ in \eqref{eq:Zd}, rather than being mixed into the compute-side service time.

\subsection{Tag-Indexed State Management}
\label{app:impl:state}
\paragraph{KV cache isolation and lifecycle.}
$M_q$ maintains a mapping from tag to KV cache.
When an $M_p^{(i)}$ starts a new decoding session, $M_q$ initializes an empty KV cache entry for its tag; during each iteration, the inference component loads and updates this KV cache to avoid recomputation of historical activations.
When the decoding session finishes, the KV cache entry associated with the tag is released to reclaim memory.
This tag-indexed management prevents cross-session contamination and enables efficient lookup and update with many active $M_p^{(i)}$.
Importantly, this isolation is necessary for session-conditioned quantities such as $\beta_i(y)$ in \eqref{eq:beta_cond} and $l_i(d)$ in \eqref{eq:tail_sum_li} to remain well-posed, where each return must advance the \emph{same} session prefix and KV state associated with that tag.

\paragraph{Extending to other per-$M_p^{(i)}$ metadata.}
The same tag-indexed design applies to other per-$M_p^{(i)}$ states maintained at $M_q$, such as speculative metadata and bookkeeping information required by the execution loop.
By enforcing strict per-tag isolation, StarSD preserves correctness under concurrency while keeping $M_q$ scalable in both memory usage and scheduling.

\section{Extended GPU Experiment Details}
\label{app:gpu_commands}

This section supplements the frequency-sensitivity study in Fig.~\ref{fig:three_ablation_panels}(\textit{b}) by verifying that GPU frequency control is indeed taking effect during our experiments. To verify whether GPU frequency control is effective, we log the measured frequency just before each draft-side inference begins in both the default setting and when the GPU graphics clock is locked.

\paragraph{GPU frequency lock command.}
We enable persistence mode and lock the graphics clock to 2400\, MHz using the following commands:
\begin{cmdbox}{GPU Frequency Lock Commands}
sudo nvidia-smi -pm 1
sudo nvidia-smi -i 0 -lgc 2400,2400
\end{cmdbox}
Unless otherwise specified, the ``Default (unlocked)'' setting means leaving the GPU in its default DVFS behavior, without applying a fixed graphics clock.

\paragraph{Frequency probing via NVML.}
After applying the frequency lock or leaving it unlocked for the default setting, we probe the instantaneous graphics clock \emph{immediately before} the draft model starts inference in each iteration. Concretely, we query NVML through \texttt{pynvml} by using the following command:
\begin{cmdbox}{GPU Frequency Measurement}
pynvml.nvmlDeviceGetClockInfo(\_gpu\_handle, pynvml.NVML\_CLOCK\_GRAPHICS)
\end{cmdbox}
We record these measurements for each experimental configuration, which means different numbers of target instances $N$. In post-processing, we discard spurious outliers where the returned frequency is very low, which can occur during handle initialization or end of string task, and then aggregate the remaining samples to produce the summary shown in Fig.~\ref{fig:gpu_frequency_comparison}.

\paragraph{Observed frequency stability.}
Figure.~\ref{fig:gpu_frequency_comparison} confirms that the locked configuration stays at 2400\,MHz across all tested $N$, whereas the default configuration exhibits higher and varying graphics clocks. This verification ensures that the trends reported in Fig.~\ref{fig:three_ablation_panels}(\textit{b}) are valid.

\begin{figure}[]
    \centering
    \includegraphics[width=0.63\linewidth]{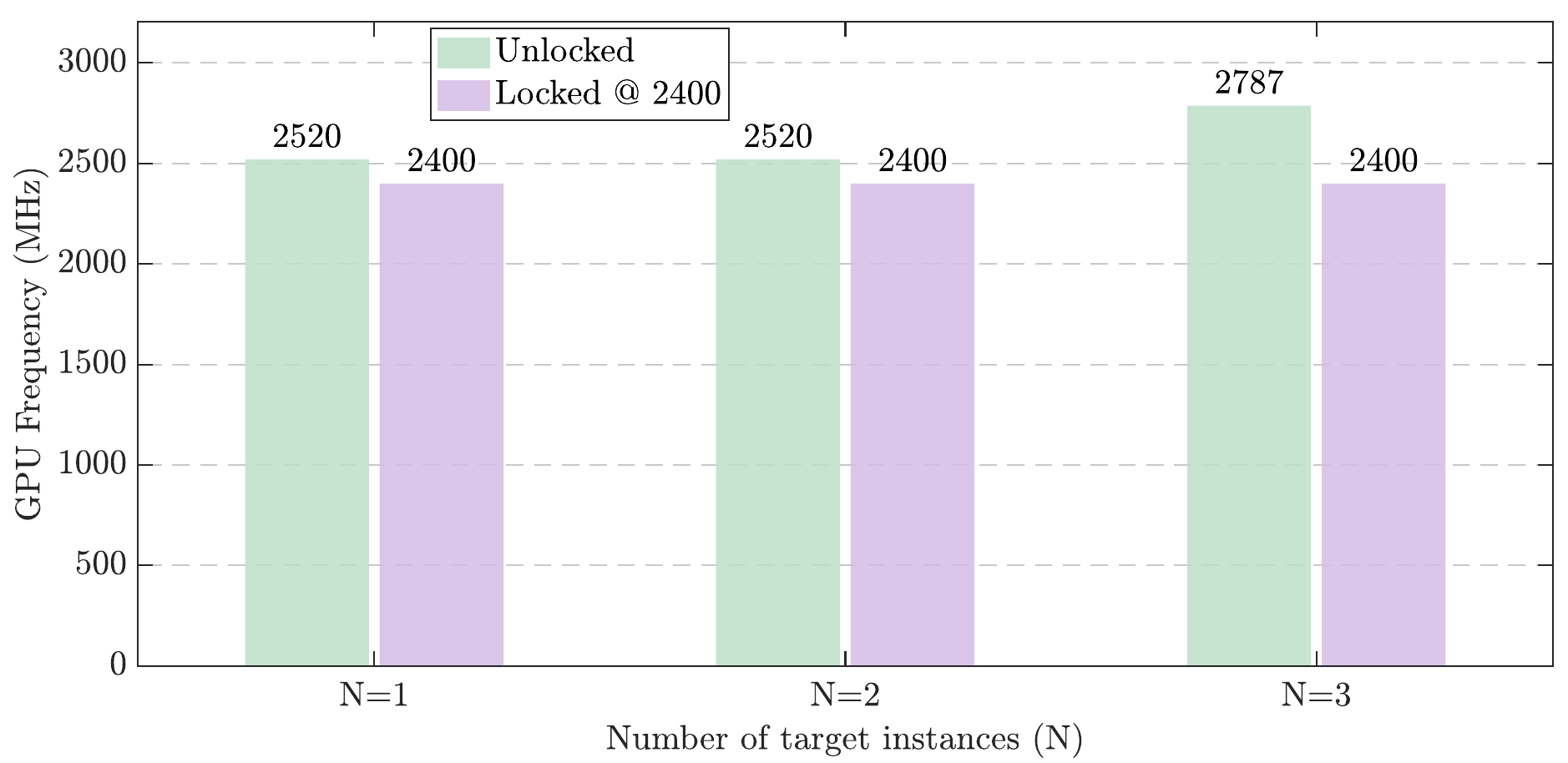}
    \caption{Measured GPU graphics clock frequency before each draft-side inference under the default DVFS setting versus a fixed 2400\, MHz lock.}
    \label{fig:gpu_frequency_comparison}
\end{figure}
\FloatBarrier

\section{Draft Inference Time Under Different Experimental Settings}
\label{app:draft_inference}
\paragraph{Draft inference speed}
\label{app:draft_infer_speed}

We report the draft-side inference time measured during our experiments.
To obtain accurate GPU kernel timing under CUDA's asynchronous execution, we use \texttt{torch.cuda.Event(enable\_timing=True)} and compute elapsed time via \texttt{event.elapsed\_time()} with an explicit \texttt{torch.cuda.synchronize()} to ensure all kernels have completed.

\begin{cmdbox}{CUDA Event Timing Snippet}
starter = torch.cuda.Event(enable_timing=True)
ender = torch.cuda.Event(enable_timing=True)

starter.record()
# draft forward pass
ender.record()

torch.cuda.synchronize()
infer_ms = starter.elapsed_time(ender)   # milliseconds
\end{cmdbox}

\paragraph{Draft inference time trend.}
Figure.~\ref{fig:draft_infer_speed_curve} plots the measured draft-side inference time under different numbers of target instances $N$ on RTX~4090 server.
As $N$ increases from 1 to 2 or above, the inference time drops markedly, where $N{=}1$ trace concentrates around $\sim$16\, ms with pronounced jitter and frequent spikes, whereas the $N{=}2$ and $N{=}3$ traces form much tighter bands around $\sim$9\, ms with only occasional outliers.

\begin{figure}[!htbp]
    \centering
    \includegraphics[width=0.75\linewidth]{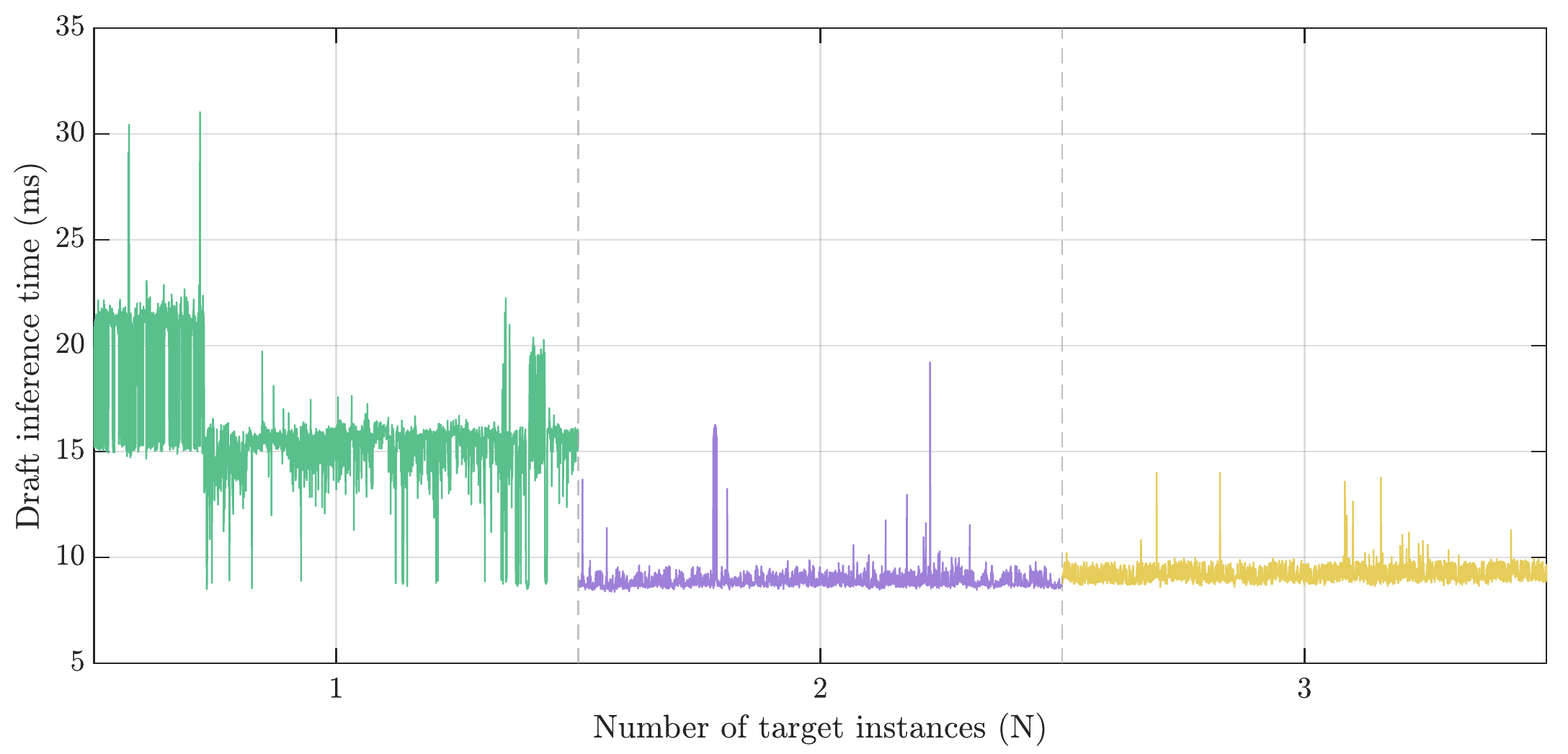}
    \caption{Draft-side inference time under different numbers of target instances $N$.}
    \label{fig:draft_infer_speed_curve}
\end{figure}

\paragraph{Summary statistics.}
Table~\ref{tab:draft_infer_speed_stats} summarizes the draft-side inference time. As $N$ increases from 1 to 2-3, the mean inference time decreases substantially and the variance drops sharply, indicating a faster and much more stable draft-side execution.

\begin{table}[!htbp]
\centering
\caption{Draft-side inference time statistics.}
\label{tab:draft_infer_speed_stats}
\begin{tabular}{lcc}
\toprule
$N$ & Mean (ms) & Variance ($\mathrm{ms}^2$) \\
\midrule
1 & 16.06 & 6.34 \\
2 & 8.86  & 0.30  \\
3 & 9.17 & 0.12  \\
\bottomrule
\end{tabular}
\end{table}

\section{Deployment Illustrations}
\label{app:deployment}

StarSD is designed to support disaggregated speculative decoding across a range of deployment topologies as shown in Fig.~\ref{fig:topology}, from intra-node multi-GPU setups to networked multi-node clusters and heterogeneous edge-server environments.
To make these supported scenarios explicit, we summarize three representative deployment modes:

\begin{itemize}
    \item \textbf{Intra-Node Multi-GPU (Intra-MG)}: \emph{machine-internal cross-GPU} within a single machine.
    \item \textbf{Inter-Node Cluster (Inter-CL)}: \emph{cross-machine} deployment over a network.
    \item \textbf{Edge-Server Hybrid (Inter-ES)}: \emph{cross-device} deployment (e.g., Jetson-PC) emphasizing heterogeneity and portability.
\end{itemize}

\begin{figure}[t]
    \centering
    \includegraphics[width=0.65\textwidth]{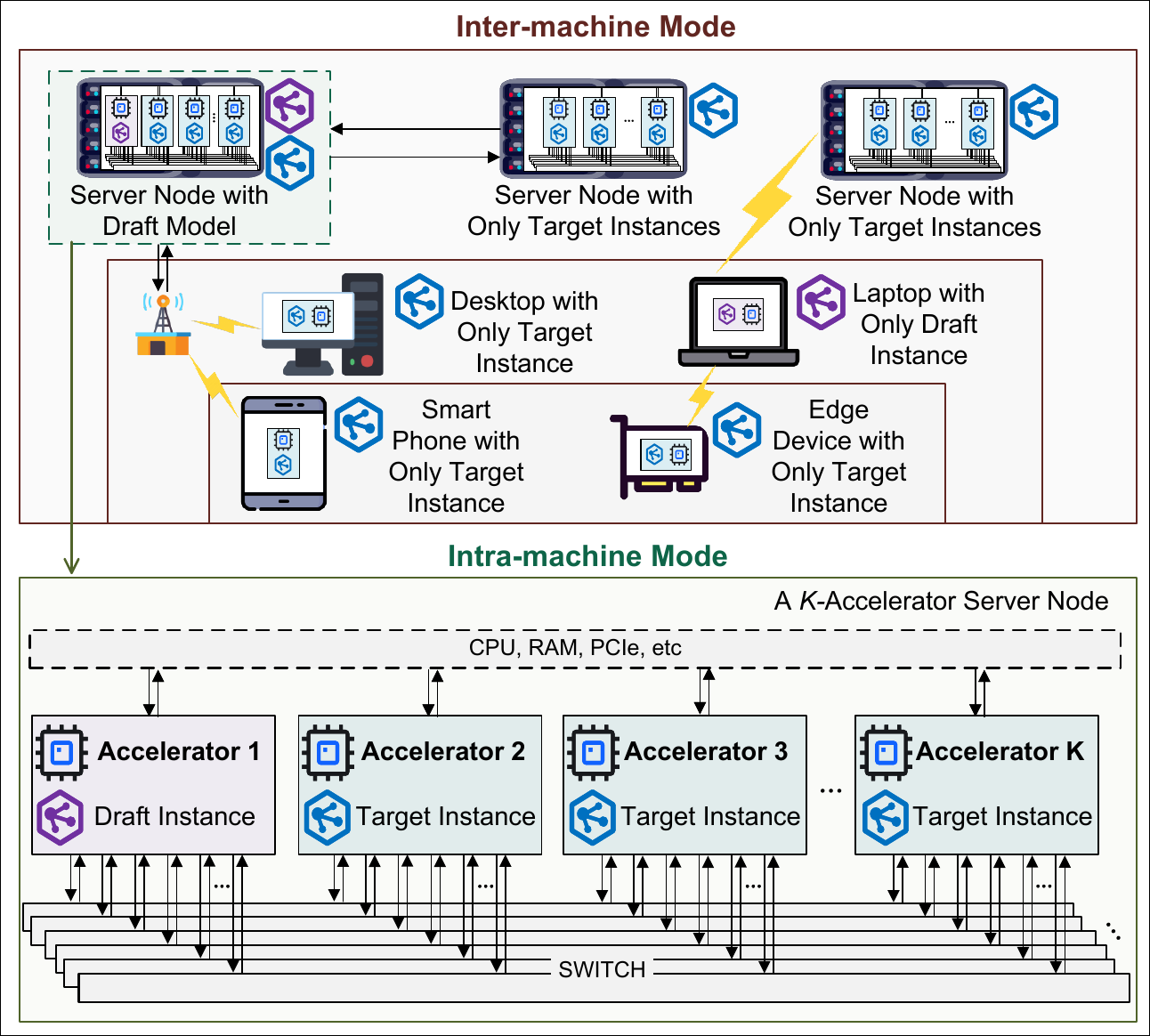}
    \caption{\textbf{Deployment illustrations for StarSD.}
    \textbf{(A) Intra-MG:} draft instance runs on one GPU and serves multiple target instances on other GPUs within the same machine, where communication can leverage intra-machine high-speed GPU interconnects, e.g., NVLink.
    \textbf{(B) Inter-CL:} draft instance runs as a draft server on one machine while target instances are placed on other machines, communication crosses the network boundary and is implemented via socket-based messaging.
    \textbf{(C) Inter-ES:} draft instance runs on one edge device, e.g., a laptop, whereas the target instances run on other edge devices, e.g., NVIDIA Jetson, and communication follows the same Inter-CL setup.}
    \label{fig:topology}
\end{figure}

\paragraph{Demo A: Intra-Node Multi-GPU (Intra-MG).}
In Intra-MG, the draft server hosting $M_q$ is placed on one GPU, while the target instances $\{M_p^{(i)}\}$ are placed on other GPUs within the same node.
A representative configuration is to place $M_q$ on a single GPU and run three target instances on the remaining GPUs of a 4-GPU server.
This topology matches the \emph{machine-internal cross-GPU} scenario and is a natural fit when co-locating $M_q$ and $M_p^{(i)}$ on the same device is infeasible due to memory pressure.
Since all communication stays within one machine, StarSD can exploit high-bandwidth intra-host GPU transport, e.g., NVLink, to reduce transfer overhead. The real-world setup is shown in Fig.~\ref{fig:real_deploy}(\textit{a}).

\paragraph{Demo B: Inter-Node Cluster (Inter-CL).}
In Inter-CL, the draft server and target instances are placed on different machines, corresponding to the \emph{cross-machine} deployment over a network.
A representative configuration is to host the draft instance alongside a subset of target instances on a single multi-GPU server, e.g., one GPU for draft and three GPUs for targets, while placing additional target instances on another server equipped with a different GPU type to form a heterogeneous cluster.
Because the dataflow must traverse the network, we use a socket-based messaging layer for portability across commodity Ethernet and heterogeneous cluster environments, while keeping the StarSD execution logic unchanged.
More specialized transports, e.g., RDMA or cluster-specific collectives, can be integrated as an optimization, but are not required for StarSD to function. The real-world setup is shown in Fig.~\ref{fig:real_deploy}(\textit{b}).

\paragraph{Demo C: Edge-Server Hybrid (Inter-ES).}
To demonstrate \emph{cross-device} portability, we also include an edge-server deployment in which one side runs on an embedded edge device, e.g., a Jetson, and the other side runs on a desktop/server GPU machine.
This scenario stresses heterogeneous compute capability and network constraints, and therefore uses the same socket-based messaging layer as Inter-CL. The real-world setup is shown in Fig.~\ref{fig:real_deploy}(\textit{c}).

\begin{figure*}[t]
    \centering
    \begin{minipage}[t]{0.19\textwidth}
        \centering
        \includegraphics[width=\linewidth]{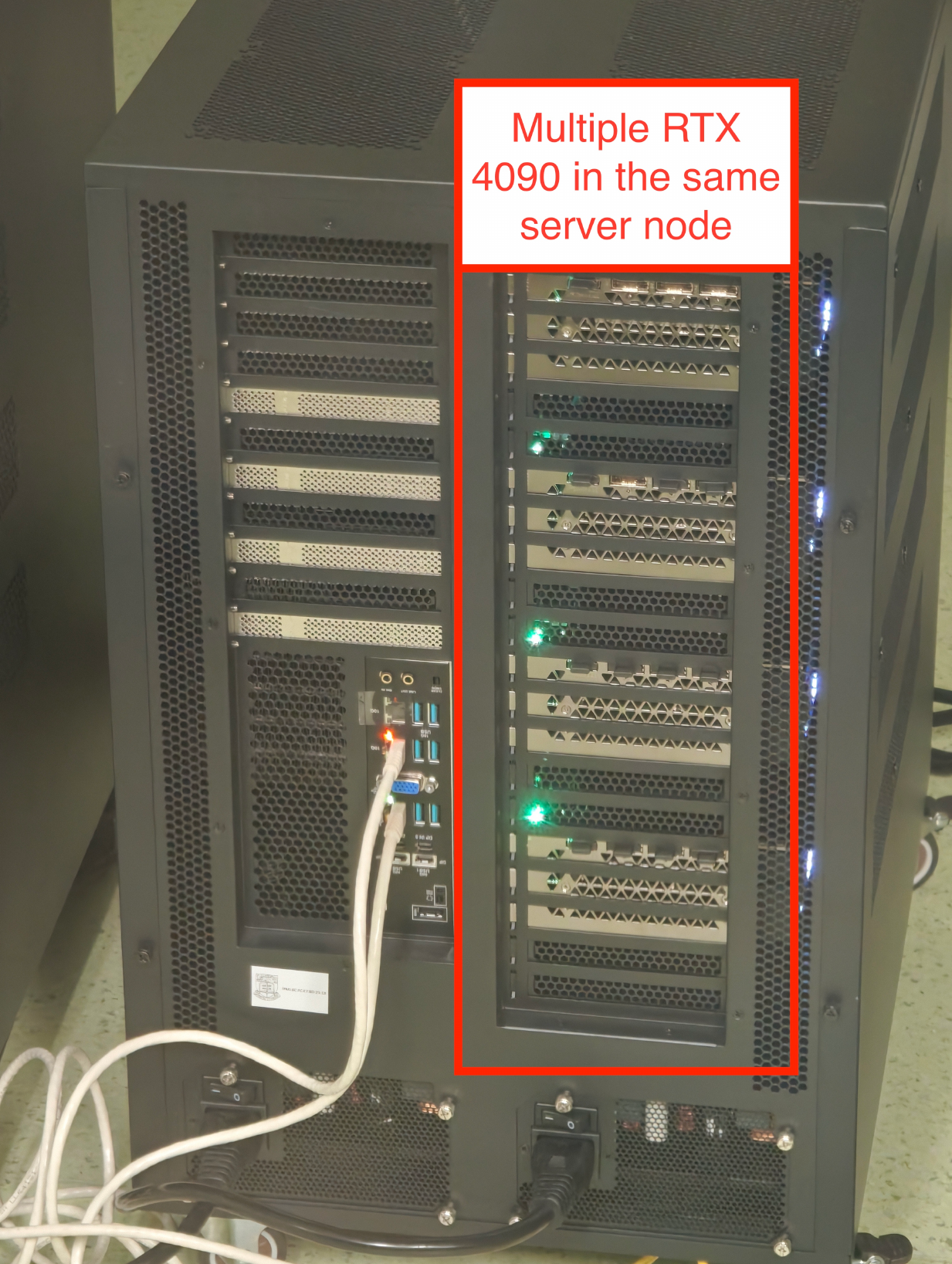}
        {\small (a) Intra-MG}
    \end{minipage}
    \hfill
    \begin{minipage}[t]{0.34\textwidth}
        \centering
        \includegraphics[width=\linewidth]{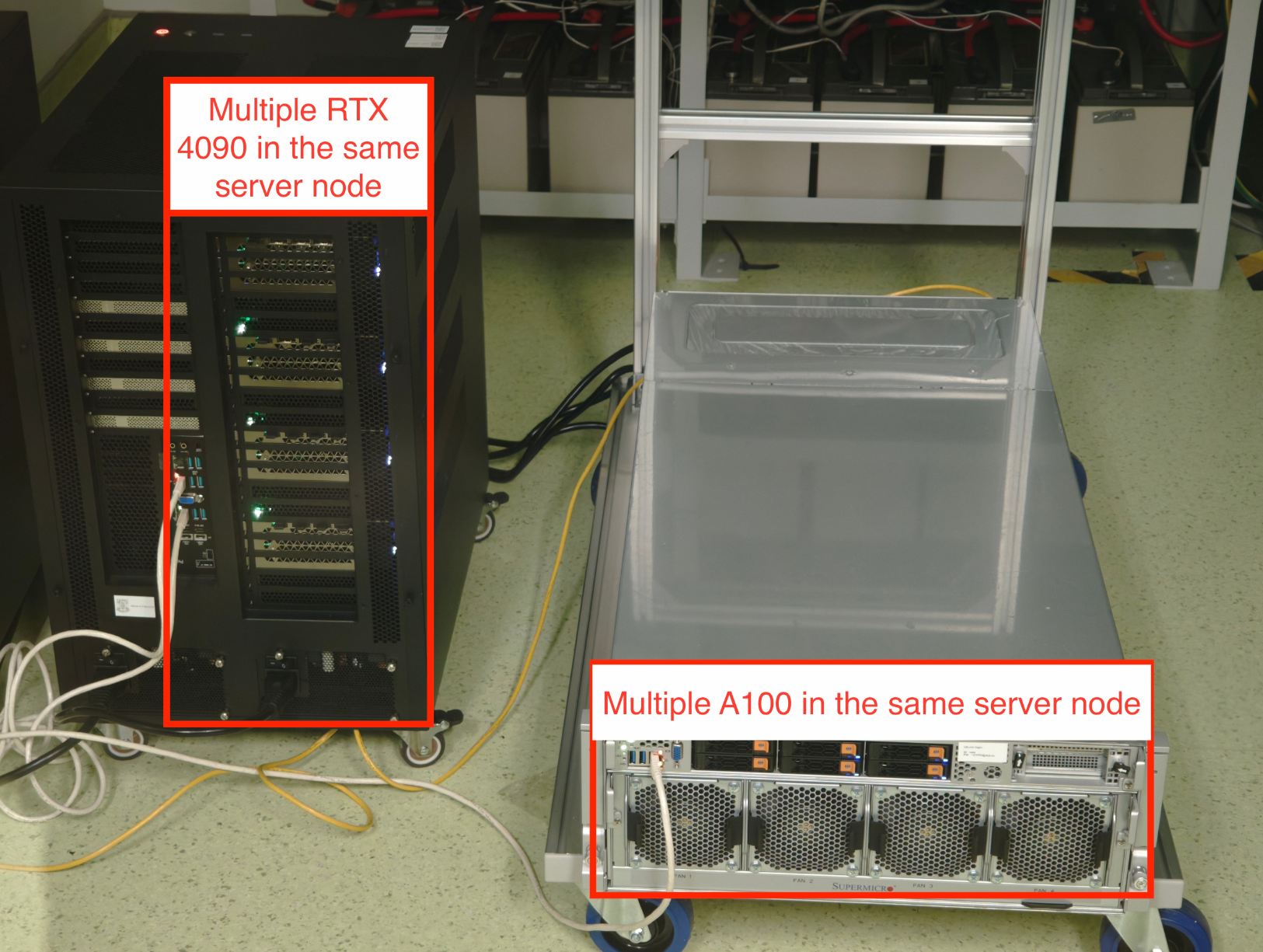}
        {\small (b) Inter-CL}
    \end{minipage}
    \hfill
    \begin{minipage}[t]{0.45\textwidth}
        \centering
        \includegraphics[width=\linewidth]{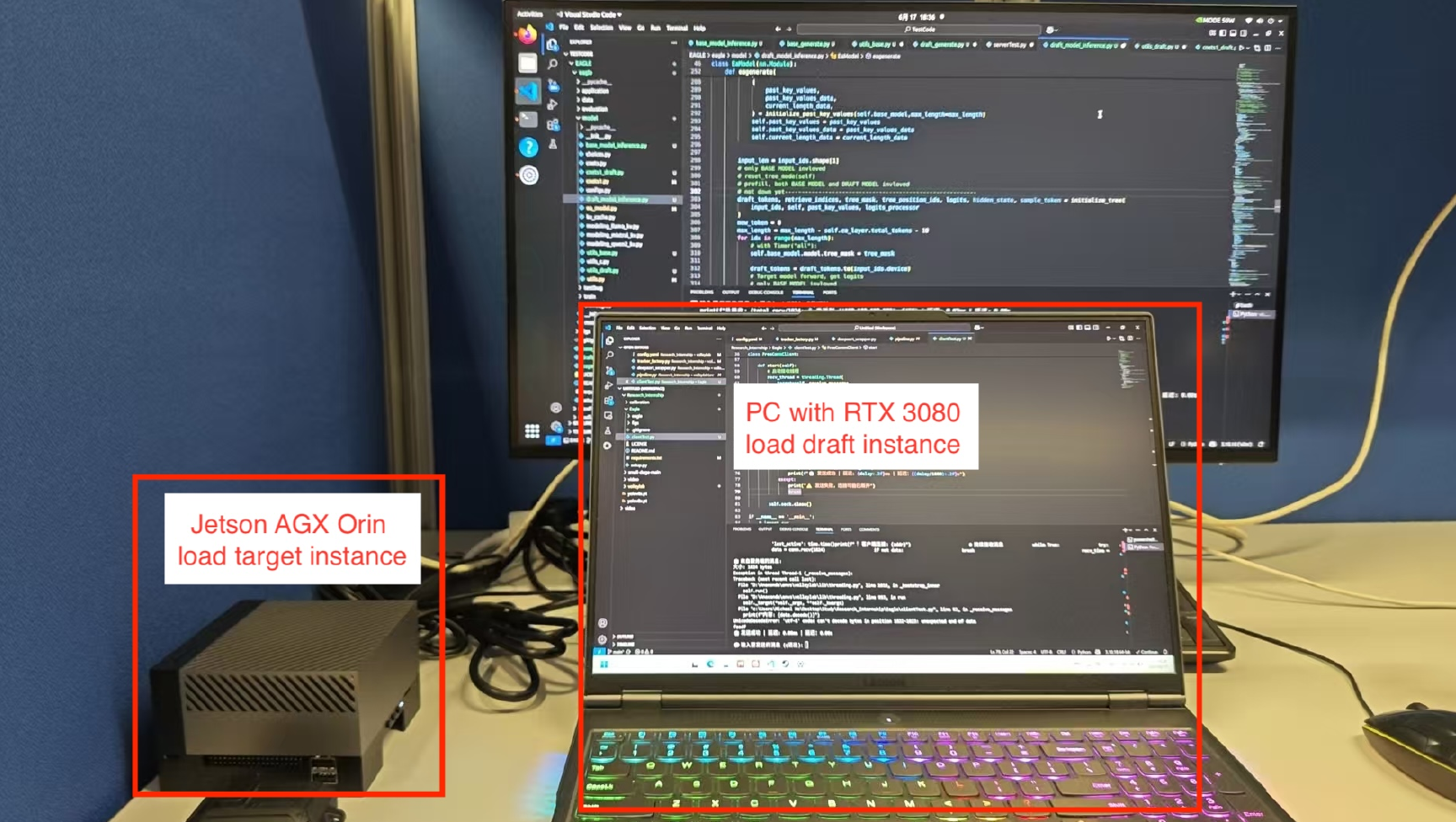}
        {\small (c) Inter-ES}
    \end{minipage}
    \caption{Photograph of the inter-server node deployment, the intra-server node deployment, and the Jetson-PC hybrid deployment.}
    \label{fig:real_deploy}
\end{figure*}

\section{Limitations and Future Directions}
\label{app:future}

\paragraph{Limitations.}
While StarSD demonstrates that a one-for-many speculative decoding framework can reduce draft-side idle gaps and improve overall utilization, the current design has several limitations.

\begin{itemize}
    \item \textbf{Fairness and per-target quality of service are not explicit goals.}
    StarSD is primarily optimized for high utilization and aggregate throughput.
    Without explicit fairness constraints, e.g., per-target rate limits, weighted scheduling, or latency service level objective (SLOs), some targets may experience worse tail latency under contention.

    \item \textbf{State isolation and memory pressure at scale.}
    Maintaining per-\texttt{tag} KV caches and speculative metadata is necessary for correctness, but it increases the memory footprint at the draft side.
    With many concurrent sessions, cache growth and fragmentation can reduce effective capacity and complicate eviction policies.

    \item \textbf{Hardware heterogeneity is under-explored.}
    When targets run on heterogeneous accelerators, e.g., different GPUs or edge devices, the variance in verification time increases, and scheduling becomes more challenging.
    The current system does not explicitly optimize for heterogeneous service rates or energy trade-offs.
\end{itemize}

\paragraph{Future directions.}
We see multiple promising directions to extend StarSD beyond the current one-for-$N$ design.

\paragraph{Toward an M-for-N draft-target fabric.}
A natural extension is an \emph{M-for-N} framework, where multiple draft instances $\{M_q^{(j)}\}_{j=1}^{N_q}$ collaboratively serve multiple target instances $\{M_p^{(i)}\}_{i=1}^{N_p}$.
Compared with one-for-$N$, the key objective is to remove the single-draft bottleneck while preserving high utilization and correctness under asynchronous returns.
Concretely, this raises several system and algorithmic questions.
(i) \textbf{Dynamic pairing and load balancing:} route each incoming verified prefix to a suitable draft instance based on queue length, predicted service time, and recent acceptance statistics.
(ii) \textbf{State placement:} decide whether per-\texttt{tag} KV caches live with a fixed draft instance, migrate across drafts, or reside in a shared cache layer.
(iii) \textbf{Consistency and replay:} ensure correctness when requests may be retried or migrated across drafts, requiring idempotent protocols and explicit versioning of speculative metadata.
(iv) \textbf{Scheduling with heterogeneity:} incorporate heterogeneous compute and network costs to achieve both high throughput and bounded tail latency.

\paragraph{Adaptive and learning-based control.}
Another direction is to make key system parameters adaptive to workload dynamics and latency SLOs, including the speculation depth $d$, batching policy, and routing strategy.
For example, the system can adjust $d$ online using lightweight predictors of acceptance length and verification latency, or learn scheduling policies that optimize a multi-objective target, e.g., throughput, latency, fairness, and energy.


\end{document}